\documentclass[11pt]{article}

\usepackage[title]{appendix}

\usepackage{amssymb}
\usepackage{amsthm}
\usepackage{amsmath,amsfonts,amssymb}
\allowdisplaybreaks
\usepackage[mathscr]{eucal}
\usepackage{exscale}
\usepackage[square,sort,comma,numbers]{natbib}
\setcitestyle{authoryear,round}
\usepackage{subfigure}

\usepackage{bm}
\usepackage{eqlist}
\usepackage[dvipsnames]{color}
\usepackage[Lenny]{fncychap}
\usepackage{multirow}
\usepackage{graphicx}
\usepackage{algpseudocode}
\usepackage{algorithm}
\usepackage{caption}
\usepackage{hyperref}
\usepackage{framed}
\usepackage{color}
\usepackage{booktabs}

\usepackage{framed}
\usepackage{algpseudocode}
\usepackage{indentfirst} 
\usepackage{longtable}

\usepackage{tabu}
\usepackage{tikz}

\usepackage{chngpage}
\usepackage{makecell}
\usepackage{threeparttable}

\usepackage[utf8]{inputenc} % this is needed for umlauts
\usepackage[english]{babel} % this is needed for umlauts
\usepackage[T1]{fontenc}    % this is needed for correct output of umlauts in pdf
\usepackage{amssymb,amsmath,amsfonts} % nice math rendering
\usepackage{braket} % needed for \Set
\usepackage{caption}
\usepackage{algorithm}

\hypersetup{
	colorlinks=true,
	linkcolor=blue,
	citecolor=blue,
	citebordercolor={1 1 0},
}

\newtheorem{theorem}{Theorem}[section]
\newtheorem{lemma}[theorem]{Lemma}
\newtheorem{proposition}[theorem]{Proposition}
\newtheorem{corollary}[theorem]{Corollary}

\newtheorem{remark}[theorem]{Remark}
\newtheorem{example}{\it Example\/}

\font\bigbf=cmbx10 scaled \magstep3

\usepackage{enumerate}
\numberwithin{equation}{section}

\begin{document}
	
	\title{\bigbf  Macroscopic fundamental diagram with volume-delay relationship: Model derivation, empirical validation and invariance property}

	\author{Ke Han$^a\thanks{Corresponding author, e-mail: kehan@swjtu.edu.cn;}$
		\quad Tao Huang$^a$
		\quad Wenbo Fan$^a$
		\quad Qian Ge$^a$
		\quad Shihui Dong$^a$
		\quad Xuting Wang$^b$
		\\\\
		$^a$\textit{\small Institute of System Science and Engineering, School of Transportation and Logistics,}
		\\
		\textit{\small Southwest Jiaotong University, Chengdu 611756, China}
		\\
		$^b$\textit{\small Institute of Smart City and Intelligent Transportation,}
		\\
		\textit{\small Southwest Jiaotong University, Chengdu 611756, China}
	}
	
	\maketitle

	\begin{abstract}		
		This paper presents a macroscopic fundamental diagram  model with volume-delay relationship (MFD-VD) for road traffic networks, by exploring two new data sources: license plate cameras (LPCs) and road congestion indices (RCIs). We derive a first-order, nonlinear and implicit ordinary differential equation involving the network accumulation (the {\it volume}) and average congestion index (the {\it delay}), and use empirical data from a 266 km$^2$ urban network to fit an accumulation-based MFD with $R^2>0.9$. The issue of incomplete traffic volume observed by the LPCs is addressed with a theoretical derivation of the observability-invariant property: The ratio of traffic volume to the critical value (corresponding to the peak of the MFD) is independent of the (unknown) proportion of those detected vehicles. Conditions for such a property to hold are discussed in theory and verified empirically. This offers a practical way to estimate the ratio-to-critical-value, which is an important indicator of network saturation and efficiency, by simply working with a finite set of LPCs. The significance of our work is the introduction of two new data sources widely available to study empirical MFDs, as well as the removal of the assumptions of full observability, known detection rates, and spatially uniform sensors, which are typically required in conventional approaches based on loop detector and floating car data. 
\end{abstract}

	\noindent {\it Keywords:} Macroscopic fundamental diagram; bathtub model; license plate cameras; congestion index; empirical analysis

\section{Introduction}
The Macroscopic Fundamental Diagram (MFD) is an aggregate model that describes the network-wide traffic dynamics in terms of space-mean density, accumulation and throughput. Such a concept is traced back to \cite{Godfrey1969}, and the existence of empirical MFD was first demonstrated by \cite{GD2008} in a field experiment carried out on a 10 km$^2$ network in Yokohama (Japan), based on loop detectors and taxi GPS data. Before that, \cite{Daganzo2007} and \cite{GD2007} have shown that such MFD can be used for modeling network-level traffic dynamics. These findings have brought the MFD to great attention in transportation research.

Since 2008, a great number of applications or extensions based on MFD have been proposed, such as perimeter control \citep{Geroliminis2012, HM2020, Keyvan2012, Keyvan2015a, Keyvan2016, SG2019, STKG2021, Hamedmoghadam2022, Tsitsokas2023, Jiang2023, Zhou2021, Zhou2023}, congestion pricing \citep{Zheng2012, DL2015, Zheng2016, Gu2018, Vickrey2020}, parking charges \citep{ZG2016, Leclercq2016}, route guidance \citep{Yildirimoglu2015, SG2017}, dynamic traffic assignment \citep{YG2014, BL2019, GFHS2019}, shared mobility \citep{Wei2020, BG2021}, bus dispatching \citep{Zhang2020, Dakic2021} and network design \citep{Loder2022}. The MFD has also inspired investigation of parsimonious models in other domains of inquiry such as pedestrian flows  \citep{SM2014, Hoogendoorn2018}, train flows \citep{Corman2019}, and aerial vehicles \citep{YYHHH2017, CM2021}. All of these studies assume that a well-defined MFD exists under some conditions. Hence, accurately and reliably estimating the MFD from empirical data is a prerequisite for the application of these methods or models.

In the literature, MFD is mainly estimated based on loop detector data (LDD) \citep{BL2009, CL2011, Ortigosa2014, Zockaie2018}, floating car data (FCD) \citep{NG2013, Du2016, Knoop2018}, or a combination of both \citep{Leclercq2014, Beibei2018, AM2016, Amhl2017, Saffari2022}. However, the LDD and the FCD  themselves have limitations that introduce errors into MFD estimates.	 As a fixed-location sensor, a loop detector cannot correctly estimate the average traffic density on the link, and a group of loop detectors bring inhomogeneous bias to the estimation of the network density \citep{BL2009, CL2011}. On the other hand, floating car data are effective in estimating the MFD only when the probe vehicles are uniformly distributed in the network, and the penetration rate is known a priori \citep{NG2013, Du2016, Knoop2018}. Although attempts are made to combine LDD with FCD to circumvent the aforementioned issues \citep{AM2016, Saffari2022}, they are either reliant on the {\it a priori} knowledge of probe penetration rate or susceptible to the uneven distribution of observers. Moreover, it has been recognized that sensor reliability brings additional challenges in real-world operations: \cite{Beibei2018} found that only 50\% of the loop detectors provide reliable counts in their case study. Finally, recent studies have introduced emerging data collection paradigms such as location based services (LBS)  \citep{Paipuri2020} and unmanned aerial vehicles (UAVs)  \citep{Paipuri2021} to the estimation of MFDs, but these technologies are still in a developing stage and not widely available. A detailed review of these approaches is presented in Section \ref{secLR}.

Aiming to address the aforementioned limitations of using LDD and FCD for MFD estimation (i.e. uniform detector distribution and known detection rates), this paper introduces a new method for estimating MFDs by exploring two data sources widely available in major cities around the globe: License plate cameras (LPCs) and road congestion indices (RCIs). The former are widely seen as traffic policy and law enforcement tools, and the latter are typically from navigation service providers such as Google, Baidu and Amap. Using these data, we propose an accumulation-based MFD \citep{Daganzo2007, GD2007} and fit it with empirical data in a 266 km$^2$ network in central Chengdu, China. This MFD model is governed by a first-order, implicit nonlinear ordinary differential equation (ODE) that involves the instantaneous network accumulation (the `{\it volume}') and the network-averaged congestion index (the `{\it delay}'). Hence, it is termed the volume-delay based MFD (MFD-VD). The main theoretical contribution and practical significance of our work are as follows:
\begin{itemize}
\item We derive the MFD-VD model for two new data sources that are widely available, with straightforward data processing and MFD fitting methods that promise good accuracy (R$^2>0.9$).

\item We theoretically prove and empirically validate that the ratio of network volume to the critical volume (corresponding to the peak of the MFD) is invariant w.r.t. the level of partial observation provided by the LPCs. This effectively removes the assumption of full observability or {\it a priori} known proportion of observed traffic, commonly required in existing studies. It is worth noting that the ratio-to-critical-value is also invariant w.r.t. the time step size, which is attributed to the same invariance principle, provided that a smaller time window is considered {\it temporally partial observation} within a larger time window.

\item  As a result, the ground-truth ratio-to-critical-value, which is an important indicator of network saturation and efficiency, can be estimated by working with a finite set of LPCs. Requirements on the spatial distribution of the LPCs are analyzed in depth, which are considerably weaker than those for loop detectors, and offer a practical guide to the application of our methods.
\end{itemize}

Finally, the practical significance is two-fold: (1) Most existing MFD modeling approaches utilize loop detector data, which are not available in many American or Asian cities, while this study offers a modeling framework based on traffic camera data, which are widely available. (2) Compared to existing approaches, the proposed MFD model does not rely on the uniform distributions of detectors, nor does it require knowledge of the detection rates, which means it can be applied to more cities and scenarios. 

The rest of the paper is organized as follows. Section \ref{secLR} provides a brief review of relevant studies in MFD estimation, Section \ref{secData} describes the LPC and RCI data used in this study. Section \ref{secMDP} derives the MFD-VD model and analyzes the invariance property as well as its sufficient conditions. Section \ref{secES} comprehensively evaluates the MFD-VD model, and validates the invariance property as well as its assumptions, followed by an impact analysis of different time step sizes. Finally, Section \ref{secConclude} provides some concluding remarks.

\section{Related work}\label{secLR}
The existing studies mainly employed loop detector data and probe vehicle data for estimating the MFD. Table \ref{MFDestimation} summarizes studies that aim to estimate the MFD using simulation or real data, along with the type of traffic data which was employed.

\begin{table}[h!]
\centering
\caption{Summary of existing studies on MFD estimation}
\label{MFDestimation}
\begin{tabular}{c cc cc}
\hline
& \multicolumn{2}{c}{Type of study} & \multicolumn{2}{c}{Type of data} 
\\
\hline
Study  &  Real-world & Simulation & LDD & FCD
\\
\hline
\cite{GD2008}  & \checkmark &  & \checkmark & \checkmark 
 \\
\cite{BL2009}   & \checkmark &   & \checkmark & 
 \\
\cite{CL2011}   &  & \checkmark  & \checkmark &  
 \\
\cite{Lu2013} & \checkmark &   & \checkmark &  \checkmark
 \\ 
\cite{NG2013}   &  & \checkmark  &  &  \checkmark
 \\
\cite{Ji2014}   & \checkmark &   &  &  \checkmark
\\
\cite{Leclercq2014}   &  &  \checkmark & \checkmark & \checkmark 
 \\
\cite{Ortigosa2014} &  & \checkmark &  \checkmark &
 \\
\cite{Saberi2014} &  & \checkmark  &  &  \checkmark
 \\
\cite{Tsubota2014} & \checkmark  &  & \checkmark  &\checkmark 
 \\
\cite{Du2016}   &  & \checkmark  &  &  \checkmark
 \\
\cite{AM2016}  &  & \checkmark  & \checkmark &  \checkmark
 \\ 
 \cite{Amhl2017}  & \checkmark &   & \checkmark &  \checkmark
 \\ 
 \cite{Amhl2018}  & \checkmark &   & \checkmark &  
 \\ 
 \cite{Dakic2018}   &  \checkmark &   &  \checkmark &   \checkmark
 \\ 
\cite{Beibei2018}   &  \checkmark &   &  \checkmark &   \checkmark
 \\ 
\cite{Knoop2018} & \checkmark &  &  & \checkmark
 \\
 \cite{Zockaie2018} &  & \checkmark & \checkmark & \checkmark
 \\
\cite{Huang2019}   & \checkmark &   &  &  \checkmark
 \\ 
\cite{Loder2019}   & \checkmark &   &  &  \checkmark
 \\  
\cite{Kavianipour2019} &  &\checkmark &  \checkmark & \checkmark
 \\
\cite{Mariotte2020}  &  & \checkmark  &  \checkmark &  
 \\
\cite{Saffari2020}   &  & \checkmark  &  \checkmark &  
 \\ 
 \cite{Amhl2021}  & \checkmark &   & \checkmark  &  
 \\ 
 \cite{Saeedmanesh2021} &  & \checkmark  & \checkmark  &
 \\
\cite{Saffari2022}   & \checkmark &  & \checkmark &  \checkmark
\\
\cite{Saffari2023}   &   & \checkmark   &  & \checkmark 
 \\\hline
\end{tabular}
\end{table}

\subsection{Loop detector data (LDD) for estimating MFDs}
	
Loop detectors, as point-based sensors, cannot accurately reflect the average link density on urban arterials, and the location of the loop detectors on the link or in the network may significantly change the measurements \citep{BL2009, CL2011}. \cite{CL2011} point out experimentally that loop detectors need to cover as many different traffic situations as possible (upstream and downstream of the traffic signals and in the middle of the road section) for the MFD be correctly estimated. A correction method is proposed by \cite{Leclercq2014}, who use the kinematic wave theory to estimate spatially averaged link density to minimize the density variance when employing loop detectors for MFD estimation. However, the improvement of accuracy is less significant in the congested parts of the network. \cite{Ortigosa2014} and \cite{Zockaie2018} propose a mathematical framework to find the optimal number of devices as well as their locations to minimize the estimation error for the MFD, nonetheless both requiring {\it a priori} knowledge of the ground-truth MFD. In a recent study, \cite{Saffari2020} apply principal component analysis (PCA) to find critical links where loop detectors should be installed, and calculate density values for all links, from which the MFD is estimated.

\subsection{Floating car data (FCD) for estimating MFDs}

As an alternative to loop detector data, floating car data from the GPS trajectories of probe vehicles provide Lagrangian observations regarding traffic states. According to \cite{Leclercq2014}, the only way to estimate the MFD without bias is to obtain information about the trajectories of all the vehicles in the network and apply Edie's definition, which is impossible in practice. A common practice is to assume known penetration rate of probe vehicles in order to extrapolate measurements to the entire network \citep{NG2013}. In addition, using floating car data to estimate the average network flow and average density requires that the probes are uniformly distributed in the network, which may not hold in reality. \cite{Du2016} propose a method for estimating the average network density consistent with Edie's general definition while addressing the non-uniform distribution of probes across the network, by assuming known penetration rate per O-D pair. \cite{Knoop2018} use a large FCD dataset provided by Google (5-minute traffic and density proxies aggregated from vehicle trajectories), which is characterized by high penetration and low inhomogeneity, to obtain  a well-defined and clear (with low dispersion) MFD for Amsterdam's network.

\subsection{LDD and FCD combined}

Some studies have explored the combination of LDD and FCD for empirical MFD estimation \citep{Leclercq2014, Beibei2018, AM2016, Amhl2017, Saffari2022, Tsubota2014}.  \cite{Leclercq2014} find through simulation experiments that combining information from FCD (average speed of the network) and LDD (average flow of the network) can provide accurate results even with a low sampling rate. But the sampling rate in their experiments is set known in advance, and the issue of non-uniform sampling is not addressed.  \cite{Beibei2018} consider both LDD and FCD, and bring to attention the issue of device reliability, where in a case study conducted in Changsha, it was found that about 25\% of the loop detectors throughout the urban network did not work at all, and only 50\% of the loop detectors provided reliable counts. This means that even under the same traffic condition, one may have different MFDs due to different sensor availability. \cite{Amhl2017} and \cite{AM2016} combine LDD and FCD, partially eliminating the critical drawbacks of these two data sources, and improve the accuracy of MFD estimation. But this approach may be susceptible to uneven distribution of detected vehicles. \cite{Saffari2022} propose a Bayesian data fusion approach to combine two traffic data sources, which relaxes the assumption of a uniform distribution, by utilizing probe vehicles traveling between specific OD pairs.

	\subsection{Other data sources for estimating MFDs}

	Emerging data sources have begun to take place in MFD studies as ICT and IoT technologies advance. \cite{Paipuri2020} propose a global framework to estimate MFD model parameters, including MFD shapes, regional trip lengths and path flow distribution, using Location Based Service (LBS) data from mobile phones. One of the distinct advantages of LBS over FCD data is that the former ensures a wide coverage of the population in the urban network, and hence higher penetration rates are observed. \cite{Paipuri2021} conducted an experiment called pNEUMA (New Era of Urban traffic Monitoring with Aerial footage) in Athens, Greece, using a swarm of UAVs to monitor more than 100,000 trajectories and their patterns in a large urban area with more than 200 road segments and 60 bus stops. \cite{Fu2020} combined four different sources of data, including GPS data of bus and car, loop detector data and transaction records of smart cards, to investigate the vehicle and passenger MFDs for cars and buses in the road network of Shenzhen.

\section{Description of empirical data}\label{secData}

Since the proposed MFD-VD model relies on two new data sources, we begin our technical discussion with the description of these data and their characteristics. The study area is located in central Chengdu, China, which is the road network within the 3rd Ring Road, spanning 266 km$^2$. The study area is under surveillance by 3,905 license plate cameras that record every passing vehicle. In addition, each of the 310 main road segments has a traffic congestion index with a time resolution of 10 min provided by Autonavi.com, a traffic navigation and digital map service provider. A larger congestion index means higher congestion or delay. Figure \ref{fignetwork} shows the locations of the 3,905 license plate cameras (left) and traffic congestion indices at 9:00 on 14 Feb.

\begin{figure}[h!]
\centering
\includegraphics[width=\textwidth]{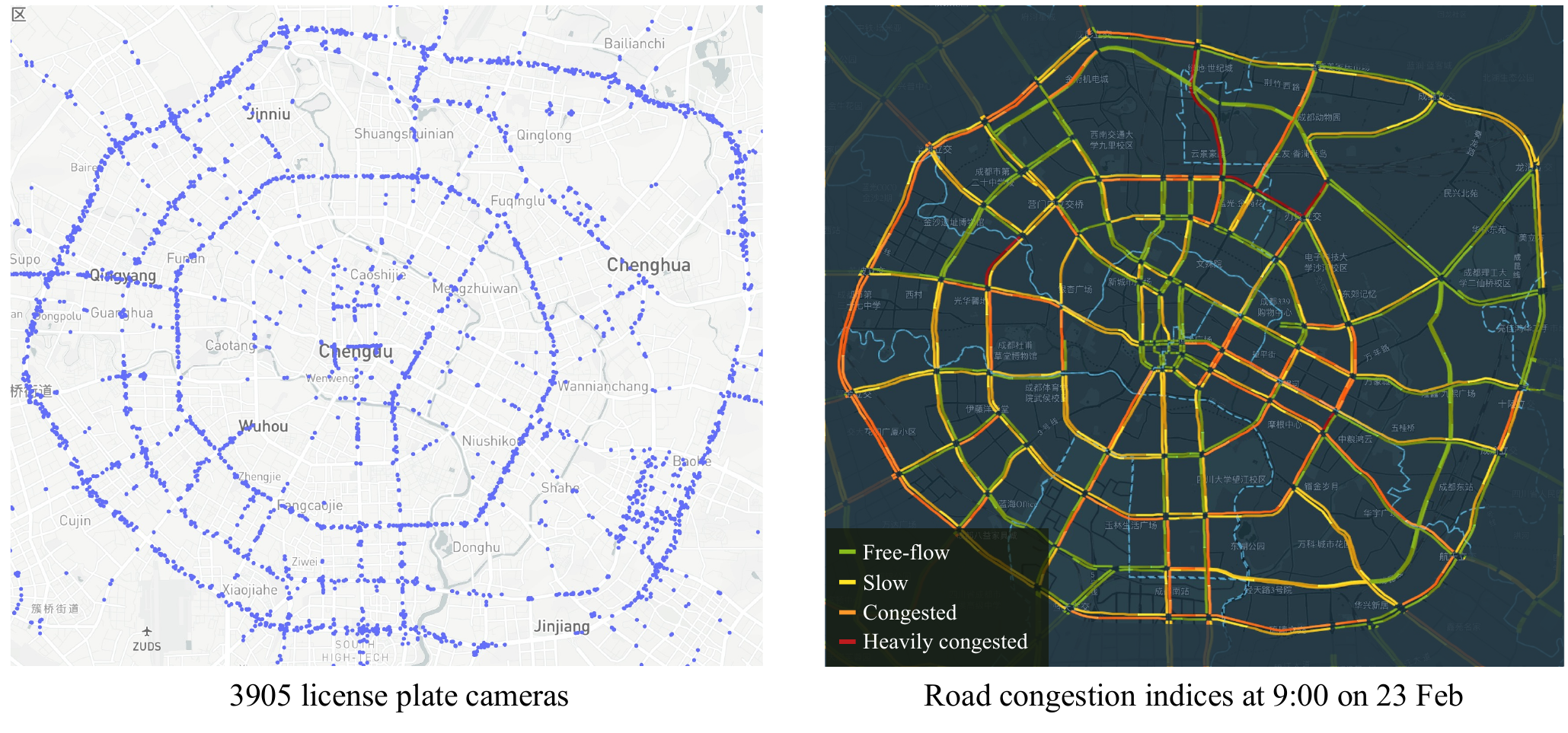}
\caption{Road network within the 3rd Ring Road in Chengdu.}
\label{fignetwork}
\end{figure}

\subsection{Network traffic volume}

This work is rooted in the notion of network traffic volume, which is characterized using license plate camera data. An elaborated treatment of such a notion is presented in Section \ref{subsecVandMFD}. Here in Figure \ref{figlongts}(a), we provide a visualization of network traffic volume observed by 1220 cameras within the study area from 1 Jan to 24 Feb, 2022. The hourly network volume is calculated as the number of distinct vehicles (license plates) captured by the 1220 cameras within one hour. Obviously, more cameras would capture more vehicles, which is confirmed later in our study (see Figure \ref{figalljan}). In fact, we make two observations: (1) The hourly network volume has a strong dependence on the camera set (in terms of both number and distribution); (2) no camera set in the real world can capture all vehicles in the network, leading to unknown proportion of those observed traffic. These observations raise considerable challenge to volume-based network modeling, which are tackled in this work.

\subsection{Road and network congestion indices (delay)}

The road congestion index (RCI), provided by Autonavi.com, is defined as the ratio between the actual travel time and the free-flow travel time of a road segment. In some literature, such a ratio is also termed {\it delay}, which will be used interchangeably with congestion index throughout this paper. Based on these indices, Autonavi categorizes the status of a road as free-flow, slow, congested and heavily congested, which are color-coded in a digital map. The correspondence between traffic status (categorical) and RCI (numerical) is illustrated in Table \ref{tabRCI}. As Autonavi only provides traffic status data, they are converted to numerical values of RCI as shown in the last row of the table. These RCI values are chosen as the median of the respective intervals; for example, $1.25=(1.0+1.5)/2$, $1.75=(1.5+2.0)/2$. Later in Section \ref{subsecvalcor}, we will experiment on different choices of these values and assess their impact on the results. 

\begin{table}[h]
\centering
\caption{Road congestion index (RCI) and traffic status}
\label{tabRCI}
\begin{tabular}{c|c|c|c|c}
\hline
      Traffic  status & Free-flow & Slow  & Congested & Heavily congested
    	\\
    	\hline
	 Color code & Green & Yellow  & Orange & Red
    	\\
    	\hline
	RCI interval & $[1.0,\, 1.5)$  & $[1.5,\, 2.0)$  & $[2.0,\, 4.0)$ & $[4.0,\, \infty)$ 
        \\
        \hline
	Value used  & \multirow{2}{*}{1.25} & \multirow{2}{*}{1.75}  & \multirow{2}{*}{3.0} & \multirow{2}{*}{5.0}
    	\\
	 in this paper & & & &
	\\
    	\hline
\end{tabular}
\end{table}

The original RCIs are updated every 10 min. The hourly RCIs are calculated as the average over these 10-min intervals. The hourly {\it network congestion index}, as shown in Figure \ref{figlongts}(b) for the period 1 Jan - 24 Feb 2022, is calculated as the mean over the 310 main road segments.

\begin{figure}[h!]
\centering
\includegraphics[width=\textwidth]{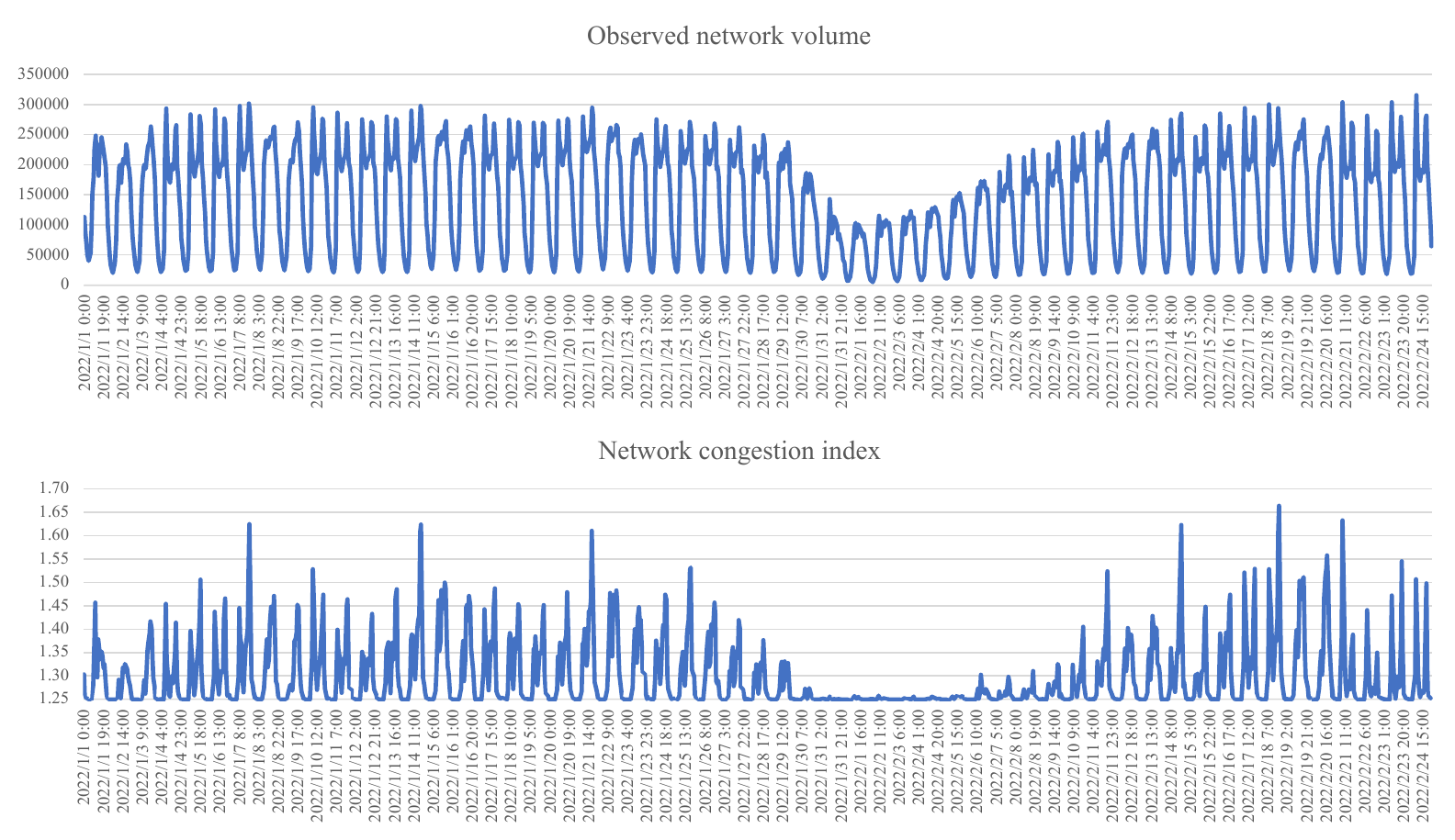}
\caption{(a) Observed hourly network volume  from 1,220 traffic cameras; (b) Network congestion index (hourly) averaged from the RCIs of 310 main road segments. The significantly lower traffic activities around 1 Feb are due to the Chinese Lunar New Year.}
\label{figlongts}
\end{figure}

\subsection{Daily traffic variations}

We primarily focus on workdays as the trip purposes and trip characteristics (e.g. O-D distribution, average trip length) are similar. By excluding weekends and national holidays (New Year and Chinese Lunar New Year), the hourly network volumes (based on the 1220 cameras) and network congestion indices (delays) on 33 workdays are statistically summarized in Figure \ref{figDV}. It can be seen that traffic (in terms of both volume and delay) fluctuates on a daily basis, to the extent indicated by the 95\% confidence intervals. In addition, the daily variations in delays are larger during morning and afternoon peaks, which suggests the presence of subtle dynamics occurring at the link level, as well as other external factors such as traffic incidents or weather that could influence traffic conditions. 

\begin{figure}[H]
\centering
\includegraphics[width=\textwidth]{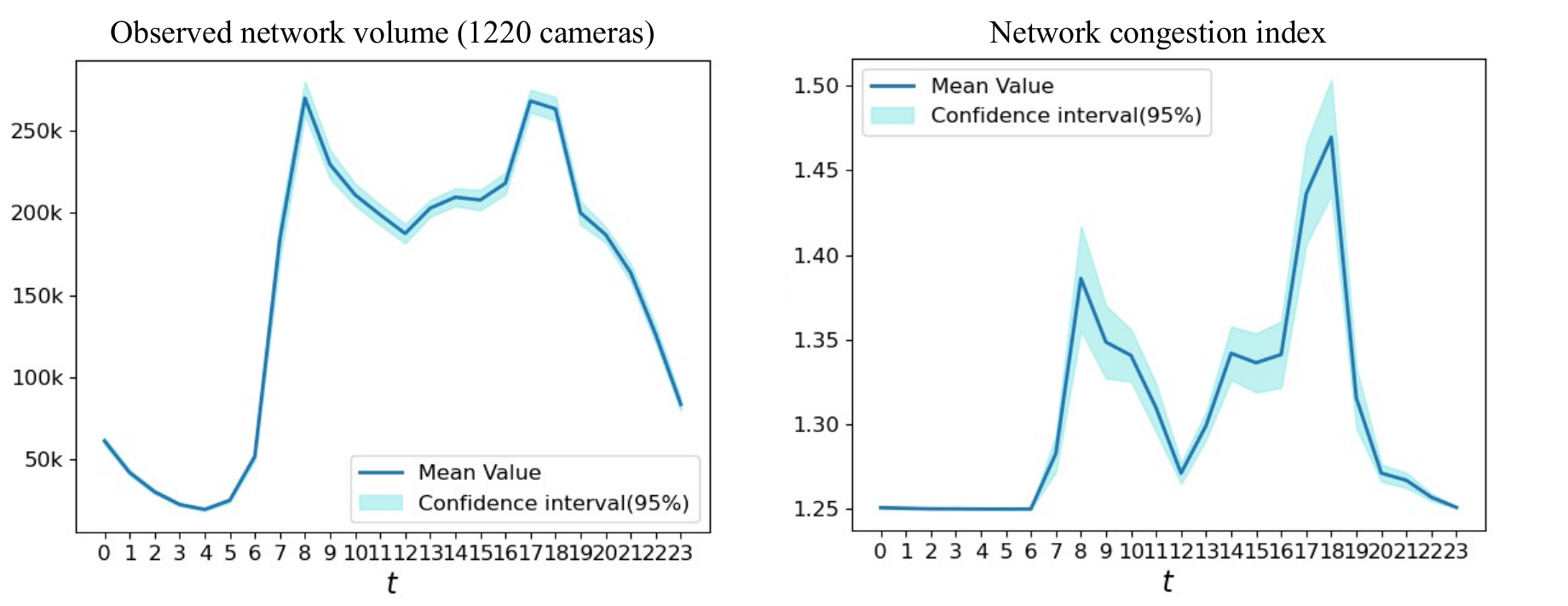}
\caption{Within-day observed network volume (left) and network congestion index (right), calculated from the 33 workdays between 1 Jan  and 24 Feb 2022.}
\label{figDV}
\end{figure}

\section{Model derivation and properties}\label{secMDP}

\subsection{Network volume and accumulation-based MFD}\label{subsecVandMFD}

In this section, time $t$ is treated as a continuous variable. The network traffic volume (accumulation) is denoted $V(t)$, which represents the number of vehicles present in the network at time $t$. While such a notion is well defined, it requires some adjustment if one attempts to interpret it using license plate camera data. Indeed, imagine that every directed link in the network is equipped with a camera, which records passing vehicles at a specific location along the link. The total number of vehicles recorded by these cameras at time instance $t$ could be any non-negative integer (even zero), despite that the network links are fully covered by those cameras. Therefore, it is only natural to define an observational interval $[t-\delta,\,t+\delta]$ to allow vehicles to pass those camera locations so that they are recorded. For partial network coverage by the cameras, which is the case in the real world, the observational interval should allow adequate sampling of network vehicles (e.g. spanning 1 hr). Given the above argument, in the rest of this paper, network volume $V(t)$ is interpreted as the number of distinct vehicles recorded within interval $[t-\delta,\,t+\delta]$ by cameras distributed ubiquitously (that is, everywhere) in the network. Such a definition is convenient when we define partially observed network volume $V_{\mathcal{C}}(t)$, using camera set $\mathcal{C}$ providing only partial coverage of the network.

We define $W(t)$ to be the rate at which traffic leave the network, and $U(t)$ to be the rate at which traffic enter the network. Here, $W(t)$ ($U(t)$) refers to both traffic exiting (entering) the network through the boundary and traffic finishing (starting) their trips within the network. The accumulation-based macroscopic fundamental diagram (MFD), following \citep{Daganzo2007}, suggests the following relationship: 
\begin{equation}\label{eqn0}
W(t)=G\big(V(t)\big),
\end{equation}
where $G(\cdot)$ is a continuous concave function, vanishing at $V=0$ and $V=V_{\text{max}}$. In the literature, the accumulation-based MFD have incorporated both constant trip lengths \citep{Daganzo2007, Mariotte2017} and time-independent negative exponentially distributed trip lengths \citep{Vickrey2020}. Other researchers have proposed the basic bathtub model \citep{AB2018, AK2022}, which assumes that all travelers have the same trip distance. Furthermore, the trip-based model \citep{LG2018} and the generalized bathtub model \citep{Jin2020} appear more flexible as they have been applied with different trip distance distributions.

\subsection{Extending link delay to network-averaged delay}
\label{section:2.1}
This part converts link-level delays (congestion indices) to network-level delay by performing link and route aggregation. Let $L$ and $R$ be the sets of links and routes in the network, respectively. We consider an arbitrary trip with route $r=\{l_1,\ldots, l_{n_r}\}$ expressed as an ordered set of links.

For an arbitrary link $l$, we define $\Phi_l\big(t_l,\,t_l+\tau_l\big)$ to be the time-averaged RCI over $[t_l,\,t_l+\tau_l]$, where $t_l$ is the departure time at link $l$ and $\tau_l$ is the (time-varying) travel time on link $l$.  The trip time on route $r\in R$ with departure time $t$ is 
\begin{equation}\label{taureqn}
\tau_r(t)=\sum_{l\in r} \text{TT}_l(t_l)=\sum_{l\in r}\text{FTT}_l \cdot \Phi_l\big(t_l,\,t_l+\tau_l\big)
\end{equation}
\noindent where $\text{TT}_l(t_l)$ denotes the travel time on link $l$ with departure time $t_l$, and $\text{FTT}_l$ denotes the free-flow travel time of link $l$. The network-wide average trip time $\tau(t)$ is defined to be the trip-weighted average:
\begin{equation}\label{eqnomegar}
\tau(t)={1\over\displaystyle\sum_{r\in R}{\omega_r}}\sum_{r\in R}\omega_r \tau_r(t)
\end{equation}
\noindent where $\omega_r$ is the weight of route $r$, which can be set as the trip volume along $r$. The following deductions are made based on \eqref{taureqn}-\eqref{eqnomegar}:
\begin{align}
\nonumber
\tau(t)&={1\over\displaystyle\sum_{r\in R}{\omega_r}} \sum_{r\in R}\omega_r\tau_r(t)={1\over\displaystyle\sum_{r\in R}{\omega_r}}\sum_{r\in R}\omega_r\sum_{l\in r}\text{FTT}_l\cdot \Phi_l\big(t_l,\,t_l+\tau_l\big)
\\
\nonumber
&\approx \Phi\big(t,\,t+\Delta t\big){1\over\displaystyle\sum_{r\in R}{\omega_r}}\sum_{r\in R}\omega_r\sum_{l\in r}\text{FTT}_l= \Phi\big(t,\,t+\Delta t\big)\cdot {1\over\displaystyle\sum_{r\in R}{\omega_r}}\sum_{r\in R} \omega_r\text{FTT}_r
\\
\label{eqndass0}
&=\Phi\big(t,\,t+\Delta t\big)\cdot \tau^0
\end{align}
\noindent where $\text{FTT}_r$ is the free-flow travel time along route $r$, $\tau^0$ is the trip-weighted average network free-flow travel time. $\Phi\big(t,\,t+\Delta t)$ is the time-averaged network congestion index over $[t,\,t+\Delta t]$
\begin{equation}\label{Phi}
\Phi\big(t,\,t+\Delta t\big)={1\over |L|}\sum_{l\in L} \Phi_l\big(t,\,t+\Delta t\big)
\end{equation}
and $\Delta t$ is a fixed constant (usually chosen to be the longest travel time in the network).

\begin{remark}
Eqn \eqref{eqndass0} suggests that $\Phi(t, t+\Delta t)$ approximates the mean of individual $\Phi_l(t_l, t_l+\tau_l)$'s, over all routes $r$, links $l\in r$, and times $t_l\in [t, t+\Delta t]$:
\begin{equation}\label{eqnproxi}
\Phi(t, t+\Delta t)\approx {{1\over\sum_{r\in R}{\omega_r}}\sum_{r\in R}\omega_r\sum_{l\in r}\text{FTT}_l\cdot \Phi_l\big(t_l,\,t_l+\tau_l\big)\over {1\over\sum_{r\in R}{\omega_r}}\sum_{r\in R}\omega_r\sum_{l\in r}\text{FTT}_l}
\end{equation}
It is worth noting that such a proxy does not require the $\Phi_l(t_l, t_l+\tau_l)$'s to be homogeneous in space or time. In fact, this approximation captures both spatial and temporal variations, as well as route choices in the network. The validity of the approximation \eqref{eqnproxi} is demonstrated in the Appendix for the target network. Finally, for \eqref{eqnproxi} to be valid, links with known RCIs should be evenly distributed in the network and include major roadways and arterials. 
\end{remark}

For notation convenience, we define $D(t)\doteq \Phi\big(t,\,t+\Delta t)$, such that \eqref{eqndass0} becomes 
\begin{equation}\label{eqndass}
\tau(t)=D(t)\cdot\tau^0
\end{equation}

\noindent We now arrive at a macroscopic relationship \eqref{eqndass}, which simply means the average network-wide trip time can be expressed as the product of $\tau^0$, the network average free-flow time, and $D(t)$, the network-averaged congestion index. Here in \eqref{eqndass}, $\tau^0$ is an unknown parameter, which will be estimated in conjunction with the MFD model based on empirical data; see Section \ref{subsecnumfitting}.

\subsection{MFD-VD model with full network volume}
\label{section:2.2}
In this section, we derive the MFD model using a volume-delay relationship (MFD-VD). The time-varying network volume $V(t)$ is defined at the beginning of Section \ref{subsecVandMFD}, which is linked to the network delay $D(t)$ via the MFD function as follows. According to flow conservation:
\begin{equation}\label{eqn1}
\dot V(t)= U(t)- W(t)= U(t)-G\big(V(t)\big)
\end{equation}
\noindent On the other hand, the flow propagation constraint, which can be also interpreted as flow conservation constraint over time, reads:
\begin{equation}\label{eqn2}
\int_{t_0}^tU(\xi)d\xi=\int_{t_0+\tau(t_0)}^{t+\tau(t)} W(\xi)d\xi\qquad \forall t
\end{equation}
\noindent where $t_0$ is some arbitrary reference time, $\xi$ is the integration dummy variable. The meaning of \eqref{eqn2} is straightforward: Traffic entering the network during $[t_0,\,t]$ leave the network during $[t_0+\tau(t_0),\, t+\tau(t)]$. Differentiating \eqref{eqn2} w.r.t. $t$ renders
\begin{equation}\label{eqn3}
U(t)=W\big(t+\tau(t)\big)\big(1+\dot \tau(t)\big)
\end{equation}

\begin{remark}
Eqn \eqref{eqn2} implies the {\it first-in-first-out} (FIFO) assumption on the network level, which has been employed in similar studies of macroscopic network dynamics \citep{AB2018, AK2022}. Note that \cite{AB2018} propose a basic bathtub model considering constant trip distances. The authors use the same formula in their paper as Eqns \eqref{eqn2}-\eqref{eqn3}. 
\end{remark}

Combining \eqref{eqn0}, \eqref{eqndass}, \eqref{eqn1} and \eqref{eqn3}, we have:
\begin{align}
\nonumber
\dot V(t)&=G\big(V(t+\tau(t))\big)\big(1+\dot \tau(t)\big)-G\big(V(t)\big)
\\
\label{eqn4}
&=G\left( V\left(t+ \tau^0 D(t)\right)\right)\left(1+\tau^0 \dot D(t)\right)-G\big(V(t)\big)
\end{align}

\noindent To circumvent the state-dependent time lag in \eqref{eqn4}, we further invoke the first-order Taylor expansion:
\begin{equation}\label{eqnTaylor}
V\big(t+\tau^0D(t)\big)\approx V(t) + \tau^0 D(t) \dot V(t)
\end{equation}
\noindent Plugging this formula back to \eqref{eqn4} renders the following first-order, implicit, nonlinear ordinary differential equation (ODE) for $V(\cdot)$\footnote{The ODE \eqref{eqn5} holds in an approximate sense due to \eqref{eqnTaylor}.}: 
\begin{equation}\label{eqn5}
\dot V(t)=G\left( V(t) + \tau^0 D(t) \dot V(t) \right) \left(1+\tau^0 \dot D(t)\right)- G\big(V(t)\big)
\end{equation}

\noindent In this ODE, $D(\cdot)$ and $\dot D(\cdot)$ are treated as known inputs because they can be estimated from the RCIs. Even if the congestion indices are not known for all road links in the network, as is our case with 310 main road links in Figure \ref{fignetwork}, $D(\cdot)$ can still be approximated as a sample mean. Therefore, full observability of the RCIs is not essential for our framework.

 \subsection{MFD-VD model with partially observed network volume}
 
 The ODE \eqref{eqn5} does not immediately admit numerical computation or empirical fitting, as it is impossible to know the true traffic volume $V(t)$. With a set $\mathcal{C}$ of license plate cameras, it is possible to partially observe the traffic volume, denoted $V_{\mathcal{C}}(t)$, which is the number of distinct vehicles captured by $\mathcal{C}$ within a unit interval centered around $t$. In this part, we work with $V_{\mathcal{C}}(t)$ in an attempt to derive the corresponding MFD-VD model with partial observation. We start by making the following critical assumption.

 \vspace{0.05 in}
 
\noindent {\bf Assumption (A1).} The total network traffic volume $V(\cdot)$ is a point-wise multiple of the volume $V_{\mathcal{C}}(\cdot)$ observed from a finite set $\mathcal{C}$ of cameras, meaning that $V(t)=k_{\mathcal{C}} V_{\mathcal{C}}(t), \forall t$ for some unknown constant $k_{\mathcal{C}}\geq 1$. 

 \vspace{0.05 in}
 
 The rationale behind this assumption is in fact quite intuitive, as we elaborate in Section \ref{secA1dis} with a discussion of conditions the LPCs should satisfy. The key point is that such conditions are considerably weaker than what was required for loop detectors, i.e. they should cover multiple traffic situations and be evenly distributed in the network \citep{CL2011, AM2016, Saffari2022}.

 Let $V(t)=k_{\mathcal{C}}V_{\mathcal{C}}(t)~ \forall t$. Provided that $V(\cdot)$ satisfies \eqref{eqn5}, we have 
\begin{equation}\label{VCODEt}
k_{\mathcal{C}} \dot V_{\mathcal{C}}=G\left(k_{\mathcal{C}}V_{\mathcal{C}}(t)+k_{\mathcal{C}}\tau^0 D(t)\dot V_{\mathcal{C}}(t)\right) \left(1+\tau^0\dot D(t)\right)-G\big(k_{\mathcal{C}}V_{\mathcal{C}}(t)\big)
\end{equation}
We define the transformed MFD function: 
\begin{equation}\label{Gc}
G_{\mathcal{C}}(x)={1\over k_{\mathcal{C}}} G(k_{\mathcal{C}}x)\qquad \forall x\in[0,\,V_{\text{max}}/k_{\mathcal{C}}],
\end{equation}

\begin{lemma}
If $G(\cdot)$ is a concave function that vanishes at $0$ and $V_{\text{max}}$, then $G_{\mathcal{C}}(\cdot)$ is also concave and vanishes at $0$ and $V_{\text{max}}/k_{\mathcal{C}}$.
\end{lemma}
\begin{proof}
The zeros of $G_{\mathcal{C}}$ are easy to verify. To show concavity, we have that $\forall x_1, x_2\in [0,\, V_{\text{max}}/k_{\mathcal{C}}]$ and $\lambda\in[0,\,1]$,
\begin{align*}
G_{\mathcal{C}}\left(\lambda x_1 + (1-\lambda)x_2\right)={1\over k_{\mathcal{C}}} G(\lambda k_{\mathcal{C}}x_1 +(1-\lambda)k_{\mathcal{C}}x_2) &\geq {1\over k_{\mathcal{C}}}\big(\lambda G(k_{\mathcal{C}}x_1)+(1-\lambda)G(k_{\mathcal{C}}x_2)\big)
\\
&=\lambda G_{\mathcal{C}}(x_1)+(1-\lambda)G_{\mathcal{C}}(x_2)
\end{align*}
\end{proof}

\noindent Eqns \eqref{VCODEt} and \eqref{Gc} show that $V_{\mathcal{C}}(\cdot)$ satisfies the following ODE:
\begin{equation}\label{VCODE}
\dot V_{\mathcal{C}}=G_{\mathcal{C}}\left(V_{\mathcal{C}}(t)+\tau^0D(t)\dot V_{\mathcal{C}}(t)\right)\left(1+\tau^0 \dot D(t)\right)-G_{\mathcal{C}}\big(V_{\mathcal{C}}(t)\big)
\end{equation}
In other words, the network volume based on partial observation $V_\mathcal{C}$ satisfies the same ODE as the full network volume $V$, except with a different MFD function $G_{\mathcal{C}}(\cdot)$.

Next, let $V^*$ be the unique critical value that corresponds to the peak of the MFD function $G(\cdot)$, then we have the following lemma
\begin{lemma}\label{cvlemma}
The critical value $V^*_{\mathcal{C}}$ corresponding to the peak of the MFD function $G_{\mathcal{C}}(\cdot)$ satisfies $V^*_{\mathcal{C}}={1\over k_{\mathcal{C}}} V^*$.
\end{lemma}
\begin{proof}
By definition:
$$
G(V^*) \geq G(V)\qquad \forall V\in[0,\,V_{\text{max}}]
$$
This is equivalent to
$$
{1\over k_{\mathcal{C}}} G\left(k_{\mathcal{C}} {1\over k_{\mathcal{C}}}V^*\right) \geq {1\over k_{\mathcal{C}}}G\left(k_{\mathcal{C}} {1\over k_{\mathcal{C}}}V\right)\qquad \forall V\in[0,\,V_{\text{max}}],
$$
that is,
$$
G_{\mathcal{C}}\left({1\over k_{\mathcal{C}}}V^*\right)\geq G_{\mathcal{C}}\left({1\over k_{\mathcal{C}}}V\right) \qquad \forall V\in[0,\,V_{\text{max}}]
$$
By changing the variables $V_{\mathcal{C}}\doteq {1\over k_{\mathcal{C}}}V$, $V^*_{\mathcal{C}}\doteq {1\over k_{\mathcal{C}}}V^*$, we have 
$$
G_{\mathcal{C}}(V^*_{\mathcal{C}})\geq G_{\mathcal{C}}(V_{\mathcal{C}})\qquad\forall V_{\mathcal{C}}\in[0,\,V_{\text{max}}/k_{\mathcal{C}}]
$$
This means that $V^*_{\mathcal{C}}={1\over k_{\mathcal{C}}} V^*$ is the unique global maximum of $G_{\mathcal{C}}(\cdot)$ over its domain, following concavity of $G_{\mathcal{C}}(\cdot)$. 
\end{proof}

An immediate consequence of {\bf (A1)} and Lemma \ref{cvlemma} is:
$$
{V(t)\over V^*}={k_{\mathcal{C}}V_{\mathcal{C}}(t)\over k_{\mathcal{C}}V_{\mathcal{C}}^*}={V_{\mathcal{C}}(t)\over V_{\mathcal{C}}^*},
$$
which leads us to the following observability-invariant property, where we define traffic observability as the (unknown) proportion of total traffic that is detected by the LPCs. 

\begin{proposition}\label{propip}{\bf (The observability-invariant property)} 
Let $V(t)$ and $V_{\mathcal{C}}(t)$ be the full network volume and partially observed volume (with camera set $\mathcal{C}$), respectively. $G(\cdot)$ is the MFD function corresponding to $V(t)$. Given that $V(t)$ and $V_{\mathcal{C}}(t)$satisfy {\bf (A1)}, then:
\begin{itemize}
\item[(1)] $V_{\mathcal{C}}(t)$ satisfies the ODE \eqref{VCODE} where $G_{\mathcal{C}}(\cdot)$ is given in \eqref{Gc}. 
\item[(2)] Their corresponding critical values satisfy $V^*=k_{\mathcal{C}}V_{\mathcal{C}}^*$.
\item[(3)]For any $t$, the ratio-to-critical-value $V_{\mathcal{C}}(t)/V_{\mathcal{C}}^*$ based on partial observation is equal to $V(t)/V^*$ with full observation, irrespective of the unknown parameter $k_{\mathcal{C}}$. 
\end{itemize}
\end{proposition}

The practical significance of Proposition \ref{propip} is that, even though the network traffic is only partially observable using prevailing sensing technology, and the proportion of those observed is unknown, one can calculate $V_{\mathcal{C}}(t)/V_{\mathcal{C}}^*$ as an approximation of the ground-truth $V(t)/V^*$, which is an important indicator of saturation level and network efficiency, by working with the observed traffic through ODE \eqref{VpODE}.

Nevertheless, the assumption {\bf (A1)} and Proposition \ref{propip} cannot be verified in reality, as the true volume $V(\cdot)$ is still unknown. Therefore, we propose the following `finite' version of {\bf (A1)}, such that it can be empirically verified (in Section \ref{subsecA1}).

\vspace{0.05 in}

\noindent {\bf Assumption (A1')} Let $\mathcal{C}_0$ and $\mathcal{C}_1$ be two sets of cameras and $\mathcal{C}_0\subset\mathcal{C}_1$, then there exists $k_{0,1}\geq 1$ such that $V_{\mathcal{C}_1}(t)=k_{0,1}V_{\mathcal{C}_0}(t)$ for all $t$.

\vspace{0.05 in}

The requirements for {\bf (A1')} to hold will be elaborated in Section \ref{secA1dis}. {\bf (A1')} is the `finite' version of {\bf (A1)} because, in case that $\mathcal{C}_1$ contains infinite number of cameras distributed ubiquitously in the network, $V_{\mathcal{C}_1}(\cdot)$ coincides with $V(\cdot)$ and {\bf (A1')} becomes {\bf (A1)}. As a result, Corollary \ref{corrip} immediately follows, which will be empirically validated in Section \ref{subsecvalcor}. 

\begin{corollary}\label{corrip}
Let $\mathcal{C}_0\subset\mathcal{C}_1$ be two sets of cameras that satisfy {\bf (A1')}, and $V_{\mathcal{C}_0}(t)$, $V_{\mathcal{C}_1}(t)$ are the corresponding partially observed volumes. $G_{\mathcal{C}_1}(\cdot)$ is the MFD function corresponding to $V_{\mathcal{C}_1}(t)$, then:
\begin{itemize}
\item[(1)] $V_{\mathcal{C}_0}(t)$ satisfies the ODE \eqref{VCODE} with the MFD function given by $G_{\mathcal{C}_0}(x)={1\over k_{0,1}}G_{\mathcal{C}_1}(k_{0,1}x)$.
\item[(2)] Their critical values satisfy $V_{\mathcal{C}_1}^*=k_{0,1}V_{\mathcal{C}_0}^*$.
\item[(3)] They yield the same ratio-to-critical-values: $V_{\mathcal{C}_1}(t)/V_{\mathcal{C}_1}^*=V_{\mathcal{C}_0}(t)/V_{\mathcal{C}_0}^*$.
\end{itemize}
\end{corollary}

\vspace{0.1 in}

\begin{example}\label{excube}{\bf An example of cubic $G(\cdot)$.}
Following \cite{GD2008}, we assume a third-order polynomial form of the MFD function:
\begin{equation}
G(V)=aV^3 + bV^2 + cV\qquad V\in[0,\,V_{\text{max}}], 
\end{equation}

\noindent where $V=V(\cdot)$ is the total network volume. Per assumption {\bf (A1)}, let $V_{\mathcal{C}}(\cdot)$ be the partially observed volume such that $V(t)=k_{\mathcal{C}}V_{\mathcal{C}}(t),~\forall t$.  By introducing the following simplified notations,
\begin{align}
\label{alphaeqn}
\alpha &\doteq V(t) + \tau^0 D(t) \dot V(t) 
\\
\label{betaeqn}
\beta & \doteq 1 + {\tau ^0}\dot D(t) 
\end{align}

\noindent ODE \eqref{eqn5} is rewritten as
\begin{align}
	\nonumber
	\dot V&=\beta G(\alpha) - G(V)
	\\
	\nonumber
	&=\beta (a \alpha^3+b \alpha^2+c \alpha)-(aV^3 + bV^2 + cV)
	\\
	\label{eqnfitting}
	&=a(\beta \alpha^3 - V^3) + b(\beta \alpha^2 - V^2) + c(\beta \alpha - V)
\end{align}

\noindent Referring to the notations \eqref{alphaeqn}-\eqref{betaeqn},  we define:
\begin{align*}
&\alpha_{\mathcal{C}}\doteq V_{\mathcal{C}}(t)+\tau^0D(t)\dot V_{\mathcal{C}}(t)={1\over k_{\mathcal{C}}}\left(V(t)+\tau^0D(t)\dot V(t)\right)={1\over k_{\mathcal{C}}}\alpha
\\
&\beta_{\mathcal{C}}\doteq 1+\tau^0\dot D(t)=\beta
\end{align*}

\noindent According to \eqref{eqnfitting}:
\begin{align}
\nonumber
k_{\mathcal{C}}\dot V_{\mathcal{C}}=\dot{V}=&a(\beta \alpha^3 - V^3) + b(\beta \alpha^2 - V^2) + c(\beta \alpha - V)
\\
\nonumber
=&  ak_{\mathcal{C}}^3(\beta_{\mathcal{C}} \alpha_{\mathcal{C}}^3-V_{\mathcal{C}}^3) +  b k_{\mathcal{C}}^2 (\beta_{\mathcal{C}} \alpha_{\mathcal{C}}^2-V_{\mathcal{C}}^2)+ c k_{\mathcal{C}}(\beta_{\mathcal{C}}\alpha_{\mathcal{C}} -V_{\mathcal{C}})
\\
\label{eqnVCODE}
\Longrightarrow \quad \dot V_{\mathcal{C}}=& a k_{\mathcal{C}}^2(\beta_{\mathcal{C}}\alpha_{\mathcal{C}}^3-V_{\mathcal{C}}^3)+b k_{\mathcal{C}}(\beta_{\mathcal{C}}\alpha_{\mathcal{C}}^2-V_{\mathcal{C}}^2)+ c(\beta_{\mathcal{C}}\alpha_{\mathcal{C}} -V_{\mathcal{C}})
\end{align}
\noindent In other words, the partial observation $V_{\mathcal{C}}(\cdot)$ satisfies the following ODE:
\begin{equation}\label{VpODE}
\dot V_{\mathcal{C}}(t)=G_{\mathcal{C}}\left( V_{\mathcal{C}}(t) + \tau^0 D(t) \dot V_{\mathcal{C}}(t) \right) \left(1+\tau^0 \dot D(t)\right)- G_{\mathcal{C}}\big(V_{\mathcal{C}}(t)\big)
\end{equation}
\noindent where 
$$
G_{\mathcal{C}}(x)= ak_{\mathcal{C}}^2x^3 + bk_{\mathcal{C}}x^2 + c x={1\over k_{\mathcal{C}}} G(k_{\mathcal{C}} x) \qquad x\in [0,\, V_{\text{max}}/k_{\mathcal{C}}]
$$

Next, we denote by $V^*$ and $V^*_{\mathcal{C}}$ the critical values that correspond to the peak of the cubic MFD functions $G(\cdot)$ and $G_{\mathcal{C}}(\cdot)$, respectively. To find $V^*$, we take the derivative and set it to zero:
$$
G'(V)=3aV^2+2bV+c=0~\Longrightarrow~V^*={-b\pm \sqrt{b^2-3ac}\over 3a}
$$
Similarly, we have 
$$
V_{\mathcal{C}}^*={-bk_{\mathcal{C}} \pm \sqrt{k_{\mathcal{C}}^2 b^2-3k_{\mathcal{C}}^2ac}\over 3ak_{\mathcal{C}}^2}={1\over k_{\mathcal{C}}}\cdot {-b\pm \sqrt{b^2-3ac}\over 3a}={1\over k_{\mathcal{C}}} V^*
$$
This leads to the invariance property
$$
{V(t)\over V^*}={k_{\mathcal{C}}V_{\mathcal{C}}(t)\over k_{\mathcal{C}} V^*_{\mathcal{C}}}={V_{\mathcal{C}}(t)\over V^*_{\mathcal{C}}}
$$

\end{example}

\subsection{Discussion of assumptions {\bf (A1)} and {\bf (A1')}}\label{secA1dis}

This section derives conditions for {\bf (A1)} or {\bf (A1')} to hold, and compares them with those for loop detectors in the literature. The following reasoning provides the underlying rationale for the validity of {\bf (A1)} and {\bf (A1')}. 

\begin{enumerate}
\item We start with a single road link $l$ covered by one or several traffic cameras. The resulting time-varying volume is denoted $V_l(t)$. We perform the following normalization over the 24-hour period $\mathcal{T}$:
$$
\tilde V_l(t)\doteq {1\over \int_{\mathcal{T}}V_l(s)ds} V_l(t)\qquad \forall t\in\mathcal{T},
$$
\noindent essentially making $\tilde V_l(\cdot)$ a  distribution over the temporal domain $\mathcal{T}$, with $\int_{\mathcal{T}}\tilde V_l(s)ds=1$. 

\item We consider a set of such temporal distributions $\{\tilde V_l(\cdot):\, l\in L\}$ obtained from a number of cameras scattered in the network. While each distribution $\tilde V_l(\cdot)$ is unique, most of them have similar shapes, characterized by morning or afternoon peaks, with low traffic during the night etc.; see Figure \ref{fig9roads} for the temporal distributions of several road segments selected from the case study. 

\begin{figure}[H]
\centering
\includegraphics[width=\textwidth]{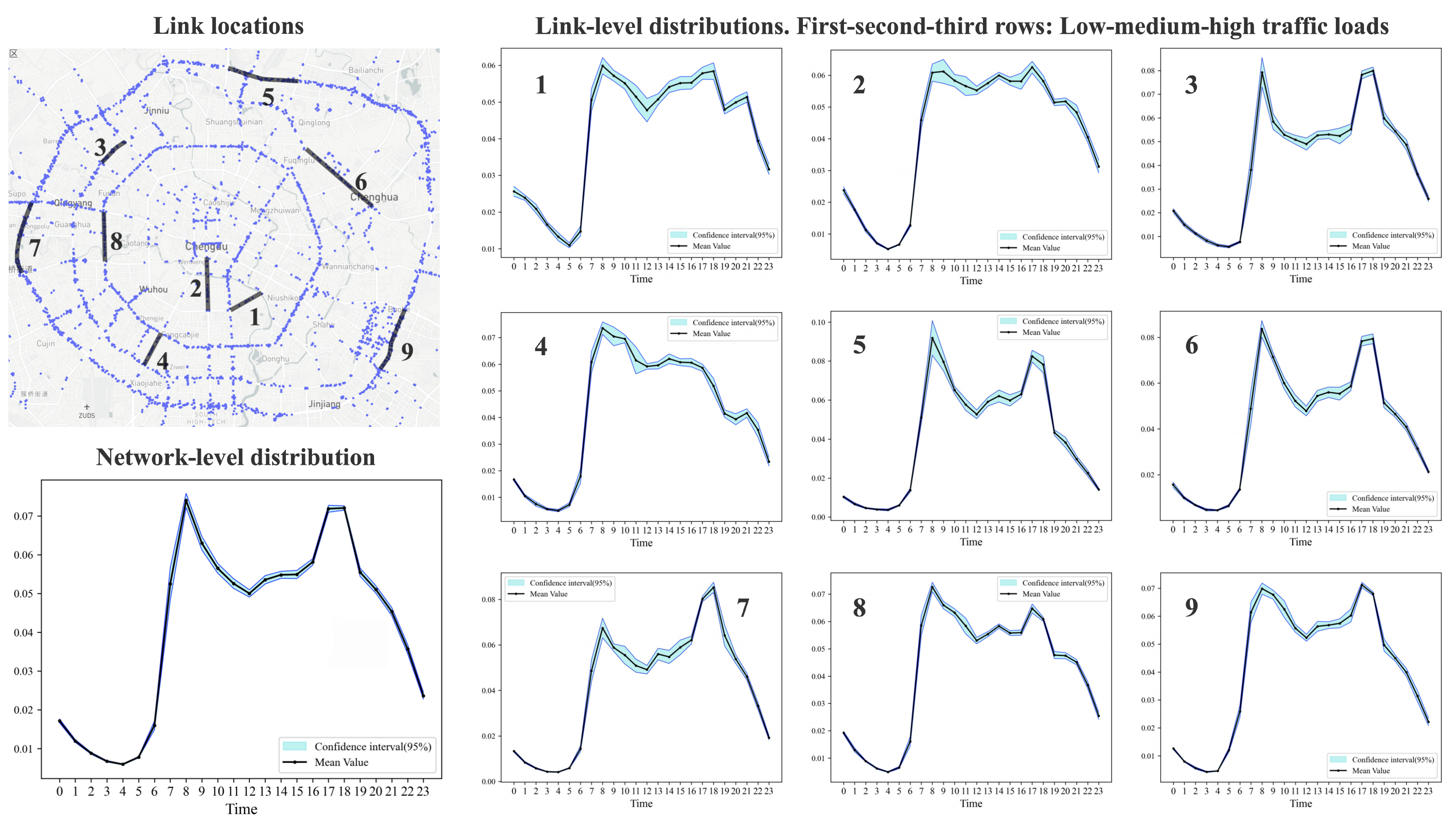}
\caption{Temporal distributions of traffic over 24-hr horizon.}
\label{fig9roads}
\end{figure}

\item At the network level with a camera set $\mathcal{C}$, we obtain the network volume $V_{\mathcal{C}}(\cdot)$ and the corresponding distribution $\tilde V_{\mathcal{C}}(\cdot)$, which is essentially a weighted average of $\{\tilde V_l(\cdot):\, l\in L_{\mathcal{C}}\}$, where $L_{\mathcal{C}}$ is the set of links covered by the cameras in $\mathcal{C}$. The weights are related to several factors including the links' traffic loads and topological configurations.

\item Let $V(\cdot)$ be the total network volume and $\tilde V(\cdot)$ the corresponding temporal distribution. If we denote by $L_{\text{all}}$ the set of all links in the network, then $\tilde V(\cdot)$ and $\tilde V_{\mathcal{C}}(\cdot)$ can be respectively interpreted as weighted average and weighted sample average of $\{\tilde V_{l}(\cdot):\,l\in L_{\text{all}}\}$. Therefore, provided that $L_{\mathcal{C}}$ is an adequate sample set\footnote{An adequate sample set should be representative of the entire population in terms of geographic locations and road types (arterial, express way), and has a size of at least 30.} of $L_{\text{all}}$, we have $\tilde V(\cdot) \approx \tilde V_{\mathcal{C}}(\cdot)$, which is equivalent to assumption {\bf (A1)}.

\item Let $\mathcal{C}_0$ and $\mathcal{C}_1$ be two adequate sample sets of $L_{\text{all}}$ and $\mathcal{C}_0\subset \mathcal{C}_1$. Then we have $\tilde V_{\mathcal{C}_0}(\cdot)\approx \tilde V_{\mathcal{C}_1}(\cdot)$, from which {\bf (A1')} follows.
\end{enumerate}

Based on the above deduction, we may derive conditions for {\bf (A1)} and {\bf (A1')} to hold. In addition to $L_{\mathcal{C}}$ being an adequate sample set of all the network links, we further notice the fact that if a link $l_1\notin L_{\mathcal{C}}$ is adjacent to other links in $L_{\mathcal{C}}$, then their temporal distributions are likely to be similar because their flows are highly correlated. Moreover, links with high traffic loads are associated with higher weights in the sample average, and therefore it helps when the camera set $\mathcal{C}$ cover high-volume links. 

In summary, the conditions for {\bf (A1)} or {\bf (A1')} to hold are: {\it The road links either covered by the camera set $\mathcal{C}$, or adjacent to those covered, are evenly distributed in the network, and contain as many high-volume links as possible}. 

\begin{figure}[H]
\centering
\includegraphics[width=.8\textwidth]{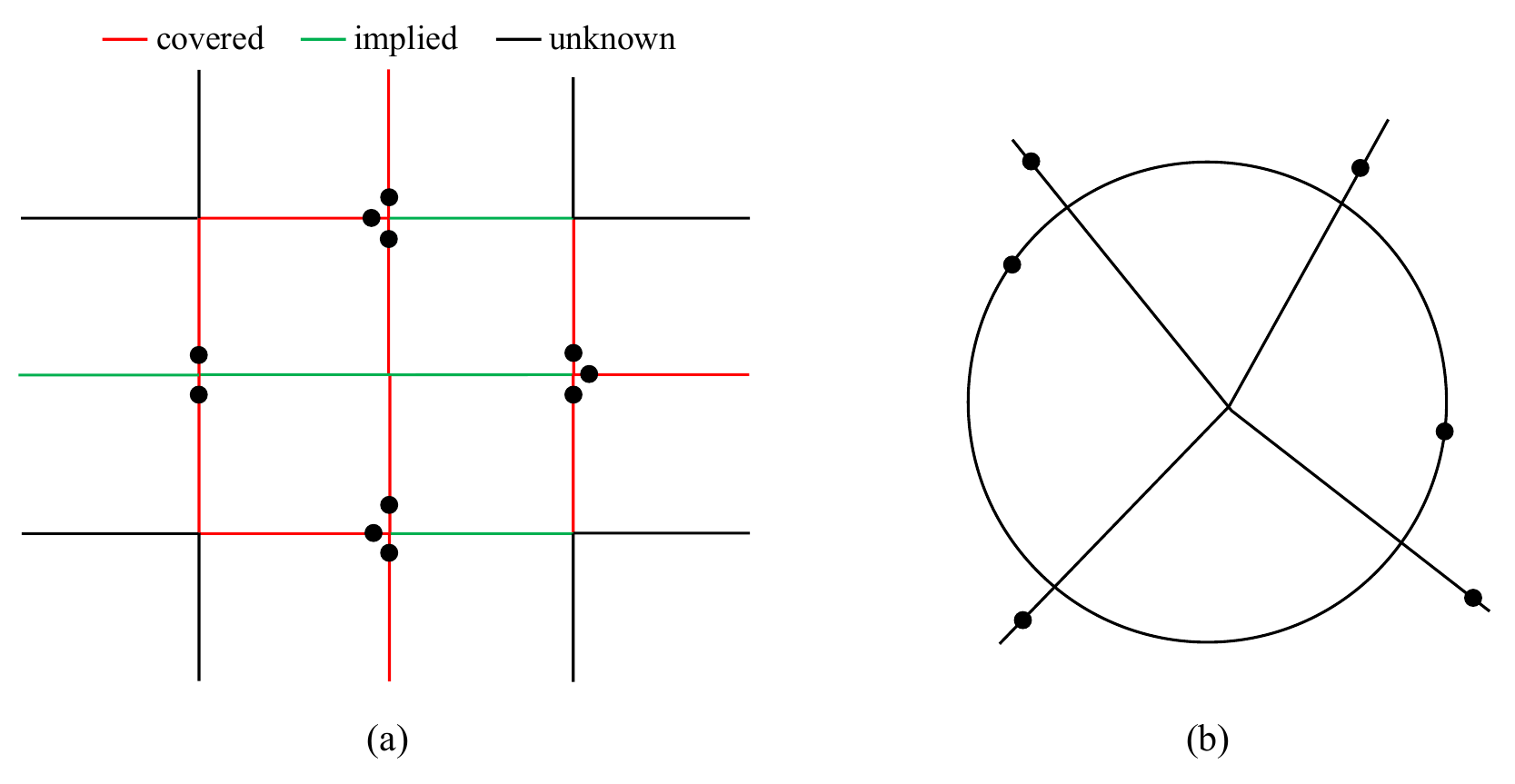}
\caption{(a): A network with cameras where links covered or implied are evenly distributed in space. (b): Even though the four links within the circle have no cameras,  cars using these links will eventually be all captured.}
\label{figA1}
\end{figure}

A loop detector measures traffic density at various parts of a road link, depending on its specific location. That means under most circumstances, it cannot accurately estimate the average density of a link, and could introduce significant bias. For example, detectors closer to the exit of a link usually produce very high density measurements, which are not representative of the overall density of the link. Therefore, average network density estimated from a number of loop detectors  may be susceptible to significant errors, unless they are evenly distributed in the network -- an assumption often made in relevant MFD studies. In contrast, in our work the cameras do not need to be evenly distributed\footnote{A license plate camera (LPC) can count the number of cars passing by within a time interval, which is much less sensitive to its specific location along the link, unlike loop detectors.}, as long as those links either covered, or implied via adjacency, are. See Figure \ref{figA1}(a) for an illustration. In addition, as $V_{\mathcal{C}}(t)$ is calculated by counting the number of distinct cars captured by all the cameras within certain period (e.g. 1 hr). The absence of cameras in one part of the network can be made up somewhere else, which has limited impact on $V_{\mathcal{C}}(t)$; see Figure \ref{figA1}(b) for an example.

\section{Empirical study}\label{secES}

\subsection{Numerical fitting of the MFD-VD with partial observation}\label{subsecnumfitting}

By virtue of the invariance property (Proposition \ref{propip}), the network volume $V_{\mathcal{C}}(\cdot)$ based on partial observation satisfies the same ODE as the full network volume $V(\cdot)$. Therefore, without causing any confusion, the subscript $\mathcal{C}$, referring to traffic volume observed from the cameras, is dropped in this subsection.

We start with a uniform time discretization $t_1, t_2, \ldots, t_s,\dots$ with increment $\Delta s$. The discrete network volume $V(t_s)$ is defined as the total number of distinct cars captured by the LPCs within $[t_s,\,t_{s+1})$. The discrete network congestion index $D(t_s)$ is defined by discretizing $\Phi\big(t,\,t+\Delta t\big)$; see \eqref{Phi}.

As in Example \ref{excube}, we assume a cubic form for the MFD function $G(x)=ax^3+bx^2+cx$, and recall from \eqref{eqnfitting} the ODE with simplified notations:
\begin{equation}\label{ODEforfit}
	\dot V=a(\beta \alpha^3 - V^3) + b(\beta \alpha^2 - V^2) + c(\beta \alpha - V)
\end{equation}
\noindent In the above, the left hand side $\dot V={d\over dt} V(t)$ can be estimated from empirical data using the centered differencing formula:
$$
\dot V(t_s) \approx {V(t_{s+1})-V(t_{s-1})\over 2 \Delta s}
$$
On the right hand side, $\beta$ can be similarly estimated by numerically differentiating $D(t)$:
$$
\beta = 1+\tau^0\dot D(t) \approx 1+ {D(t_{s+1})-D(t_{s-1})\over 2\Delta s}\tau^0
$$

We note from \eqref{ODEforfit} that the fitting is linear in $a$, $b$ and $c$, but  nonlinear in $\tau^0$. In fact, this formula contains $\tau^0$ to the 4-th power. To retain tractability of linear regression and investigate the impact of $\tau_0$ on the goodness-of-fitting,  we conduct an enumeration of $\tau^0$ and perform multi-variate linear regression for $a, b, c$ with given $\tau^0$. This allows us to easily find the four-tuple $(a,\,b,\,c,\,\tau^0)$ that best fit the empirical data. 

Regarding possible values of $\tau^0$, we note that since the travel purposes and trip characteristics change over the course of a day, $\tau^0$ may also vary with time. This is also reflected in \eqref{eqndass}, where $\tau^0$ is defined to be the trip-averaged free-flow time, which implicitly depends on origin-destination demands and trip volumes. Consequently, we consider the following four periods in Table \ref{tabperiods} and assume constant free-flow trip time $\tau^0_i$ in each period $i=1,2,3,4$; that is, $\tau^0=(\tau^0_1, \tau^0_2, \tau^0_3, \tau^0_4)$. 
\begin{table}[h]
\centering
\caption{The four time intervals of a day used to describe free-flow trip time $\tau^0$}
\label{tabperiods}
\begin{tabular}{c|c|c|c|c}
\hline
No. & 1 & 2 & 3 & 4
\\\hline
Time interval & 7:00-9:00 & 10:00-15:00 & 16:00-19:00 & 20:00-6:00
\\\hline
Description & Morning peak & Off-peak (day) & Afternoon peak & Off-peak (night)
\\\hline
\end{tabular}
\end{table}

In this section, we use empirical data from 33 workdays between 1 Jan and 24 Feb to fit the model. The total number of license plate cameras that are consistently functional during these 33 days is 1220, which are used to obtain the traffic volume $V(s),\, s=1,\ldots, 24$ with $\Delta s=1$ hr (tests on smaller time steps $\Delta s$ will be elaborated in Section \ref{subsectimegra}).

Here, two different datasets are considered for the multilinear regression \eqref{ODEforfit}:
\begin{itemize}
\item[(1)] Hourly data from 33 days: $\Omega_1 \doteq \big\{V^d(s),\, D^d(s): d=1,\ldots, 33;\, s=1,\ldots, 24\big\}$;
\item[(2)] 33-day averaged hourly data: $\Omega_2 \doteq \big\{\bar V(s),\, \bar D(s): s=1,\ldots, 24 \big\}$ where 
$$
\bar V(s)={1\over 33}\sum_{d=1}^{33} V^d(s),\quad \bar D(s)={1\over 33}\sum_{d=1}^{33} D^d(s)
$$
\end{itemize}
\noindent The second dataset $\Omega_2$ reduces daily variations by aggregating the hourly data over multiple days.

We consider the following goodness-of-fit indicators. Let $\{\hat y_i:  i=1,\ldots, n\}$ be the estimated values and $\{y_i: i=1,\ldots, n\}$ be the actual values:
\begin{itemize}
\item Standard error, also known as root mean square error (RMSE)
$$
\text{StdErr}=\sqrt{{\sum_{i=1}^n(\hat y_i - y_i)^2\over n}}
$$
\item R-squared
$$
R^2=1- {\sum_{i=1}^n (\hat y_i - y_i)^2 \over \sum_{i=1}^n (y_i - \bar y)^2}
$$
\noindent where $\bar y$ is the mean of $\{y_i: i=1,\ldots, n\}$.
\item Symmetric Mean Absolute Percentage Error
$$
\text{SMAPE}={1\over n}\sum_{i=1}^n { 2|\hat y_i -y_i| \over |\hat y_i|+|y_i|}\times 100\%
$$
\end{itemize}

By enumerating all components $\tau^0_i$ ($i=1,2,3,4$) ranging from 0.1 to 1.0 (hr) with an increment of 0.1 hr, we find the vector that yields the minimum standard error in both $\Omega_1$ and $\Omega_2$, which is $\tau^{0,*}=(0.4, 0.4,0.3,0.3)$ (hr). Table \ref{tabKPIs} compares various fitting performance measures of a few choices of $\tau^0$. It can be seen that $(0.4, 0.4,0.3,0.3)$ also performs well in terms of R$^2$ and SMAPE. 

%\begin{remark}
%The network congestion index and free-flow time appear in the form of a product $\tau^0 D(t)$ or $\tau^0 \dot D(t)$ in the ODE \eqref{VCODE}, which means the fitted MFD depends only on their product, not individual values. 
%\end{remark}

\begin{table}[h!]
\centering
\caption{Goodness-of-fit performance measures for a few choices of $\tau^0$.}
\label{tabKPIs}
\begin{tabular}{c| ccc| ccc}
\hline
\multirow{2}{*}{$\tau^0$ (hrs)}  & \multicolumn{3}{c|}{Original data $\Omega_1$} & \multicolumn{3}{c}{Daily averaged data $\Omega_2$} 
\\
\cline{2-7}
  &  StdErr & R$^2$ & SMAPE & StdErr & R$^2$ & SMAPE 
\\
\hline
${\bf (0.4,0.4,0.3,0.3)}$  & 2968 & 0.916 & 16.06\% & 1316 & 0.962 & 11.87\%
 \\
 \hline
 $(0.4,0.4,0.4,0.4)$  & 4553 & 0.871 & 21.02\% & 3400 & 0.902 & 14.71\%
 \\
 \hline
  $(0.3,0.3,0.3,0.3)$  & 4397 & 0.876 & 20.81\% & 3217 & 0.908 & 14.47\%
 \\
 \hline
 $(0.5,0.5,0.4,0.4)$  & 3041 & 0.914 & 16.06\% & 1335 & 0.962 & 11.22\%
 \\
 \hline
 $(0.5,0.6,0.4,0.5)$  & 3694 & 0.896 & 16.57\% & 2121 & 0.939 & 11.86\%
 \\
 \hline
 $(0.3,0.3,0.2,0.2)$  & 3447 & 0.903 & 18.88\% & 2191 & 0.937 & 16.60\%
 \\\hline
\end{tabular}
\end{table}

In Figure \ref{figfitting} we visualize the fitting performances on $\Omega_1$ (left) and $\Omega_2$ (right). From this figure we see that the daily variation of $\dot V(\cdot)$ is relatively significant during 6:00-8:00, 9:00-10:00, 16:00-17:00, and 19:00-20:00. This is because traffic volume undergoes considerable temporal variations during these periods (see Figure \ref{figDV}), causing large variations in the derivative $\dot V(\cdot)$. Another observation is that the proposed model fits well the empirical data, on both disaggregated ($\Omega_1$) and aggregated ($\Omega_2$) levels.  

\begin{figure}[h!]
\centering
\includegraphics[width=\textwidth]{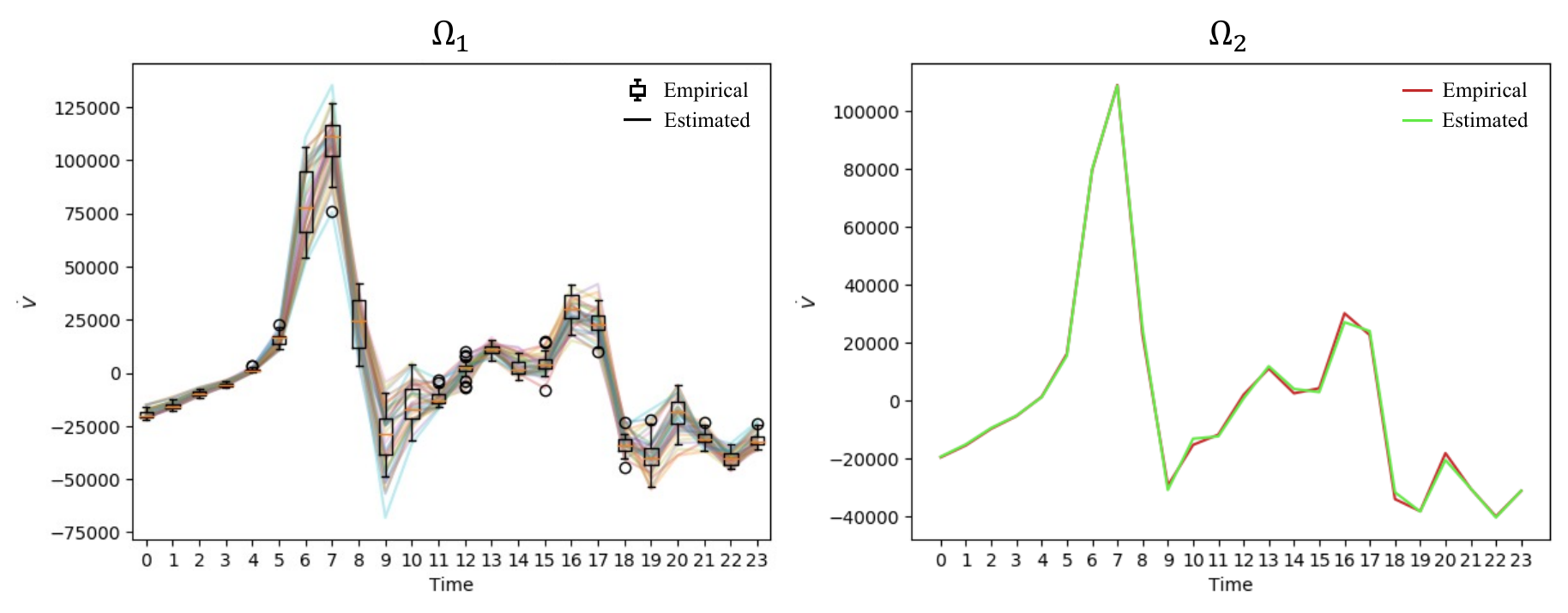}
\caption{Visualization of the fitting results with $\tau^{0,*}=(0.4,0.4,0.3,0.3)$. Left: fitting results on $\Omega_1$ where the box plots represent empirical data and the solid lines are estimated by our method, one for each day. Right: fitting results on $\Omega_2$ where both curves are averaged over 33 days.}
\label{figfitting}
\end{figure}

Figure \ref{figMFDs} shows the resulting MFD curves obtained from fitting the datasets $\Omega_1$ and $\Omega_2$, with the critical volume being $V^*=407715$ and $447534$, respectively. From Figure \ref{figlongts}, which shows hourly network volume calculated using the same 1220 cameras, we see that the highest hourly volume is around 300k, which means that traffic remained in the left branch of the MFD. 

\begin{figure}[h!]
\centering
\includegraphics[width=\textwidth]{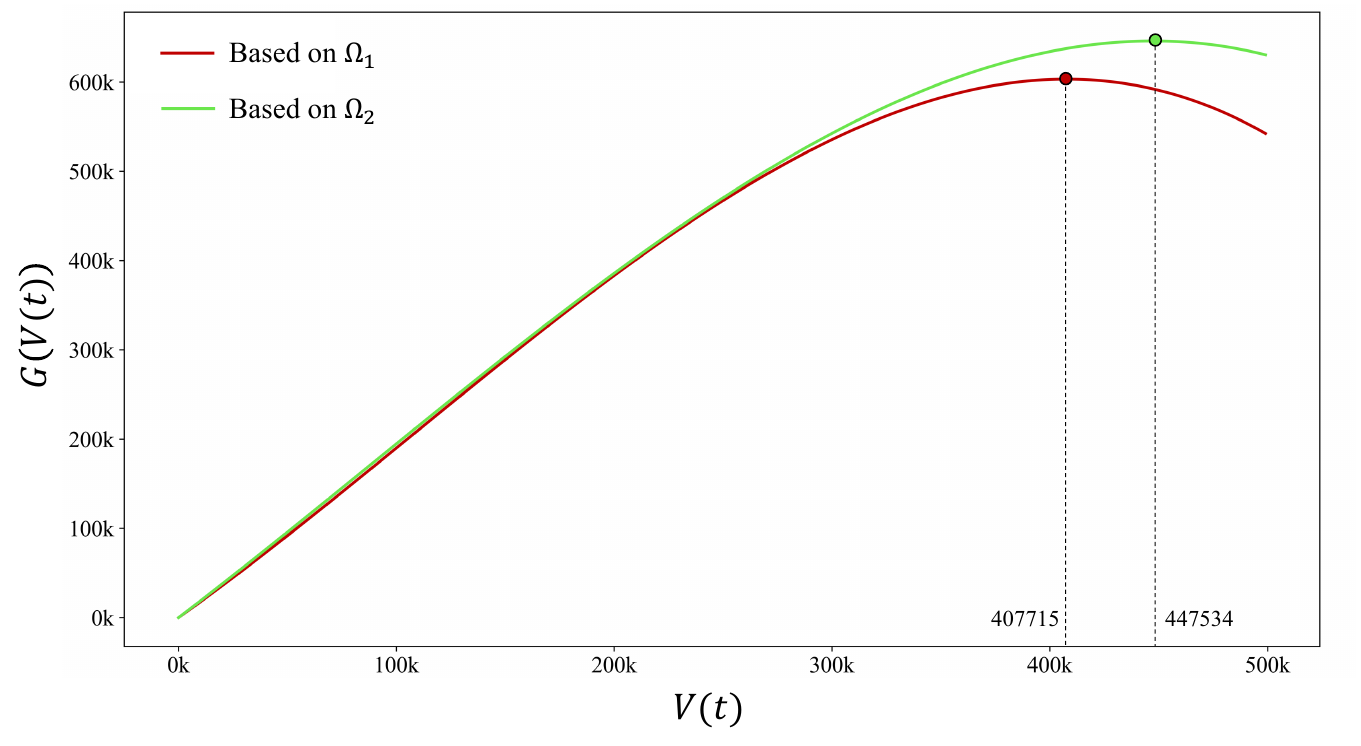}
\caption{MFD curves fitted based on $\Omega_1$ (33 days of hourly data) and $\Omega_2$ (daily averaged data).}
\label{figMFDs}
\end{figure}

\subsection{Observability-invariant property of the MFD-VD model}
	A hallmark of the MFD-VD model is that it defines the network accumulation as the volume of traffic captured by license plate cameras within certain time interval. While this provides a technologically viable way to monitor the network state as traffic cameras are densely available in many cities, they can only capture an unknown portion of the total traffic present in the network. This section uses empirical data to validate the observability-invariant property of the MFD-VD model.

\subsubsection{Empirical validation of Assumption (A1')}\label{subsecA1}

We use empirical data to validate assumption {\bf (A1')} as well as Corollary \ref{corrip}. To do this, we need to create a hierarchy of traffic observability by generating different camera sets $\mathcal{C}_i$. We make use of the fact that the set of license plate cameras in the study area varies on a daily basis, caused by regular maintenance activities and/or random data transmission failures. The following procedure is followed to achieve this goal. 
\begin{enumerate}
\item Select an arbitrary set of workdays, denoted $\mathcal{D}_0$, and determine the set of cameras $\mathcal{C}_0$ that are fully functional during $\mathcal{D}_0$.

\item Choose a subset of workdays $\mathcal{D}_1\subset \mathcal{D}_0$, and determine the set of cameras $\mathcal{C}_1$ that are fully functional during $\mathcal{D}_1$; clearly $\mathcal{C}_0 \subset \mathcal{C}_1$.

\item Repeat step 2 until a sequence of sets $\mathcal{D}_n\subset \ldots \subset \mathcal{D}_1\subset \mathcal{D}_0$ are found with $\mathcal{C}_0\subset \mathcal{C}_1\subset\ldots\subset \mathcal{C}_n$ for some $n\geq 1$.

\item Output the set $\mathcal{D}_n$ with a range of traffic observabilities characterized by $\mathcal{C}_0,\ldots, \mathcal{C}_n$. 
\end{enumerate}

We take $\mathcal{D}_0$ to be the 33 workdays between 1 Jan and 24 Feb, with $|\mathcal{C}_0|=1220$. Then, we let $|\mathcal{D}_1|=19$ with $|\mathcal{C}_1|=1760$ and $|\mathcal{D}_2|=8$ with $|\mathcal{C}_2|=2161$. This gives us a test period spanning 8 workdays ($\mathcal{D}_2$), with three levels of observability: $|\mathcal{C}_0|=1220$, $|\mathcal{C}_1|=1760$, and $|\mathcal{C}_2|=2161$. The spatial distributions of $\mathcal{C}_i$'s are shown in Figure \ref{figcamera3p}. Obviously, there are other ways to sample the cameras to form various subsets. The procedure described above is chosen because (1) it is based on clear instructions and easy to follow; and (2) it ensures that the selected cameras are always functional during the study period, providing valid data for processing.

\begin{figure}[h!]
\centering
\includegraphics[width=\textwidth]{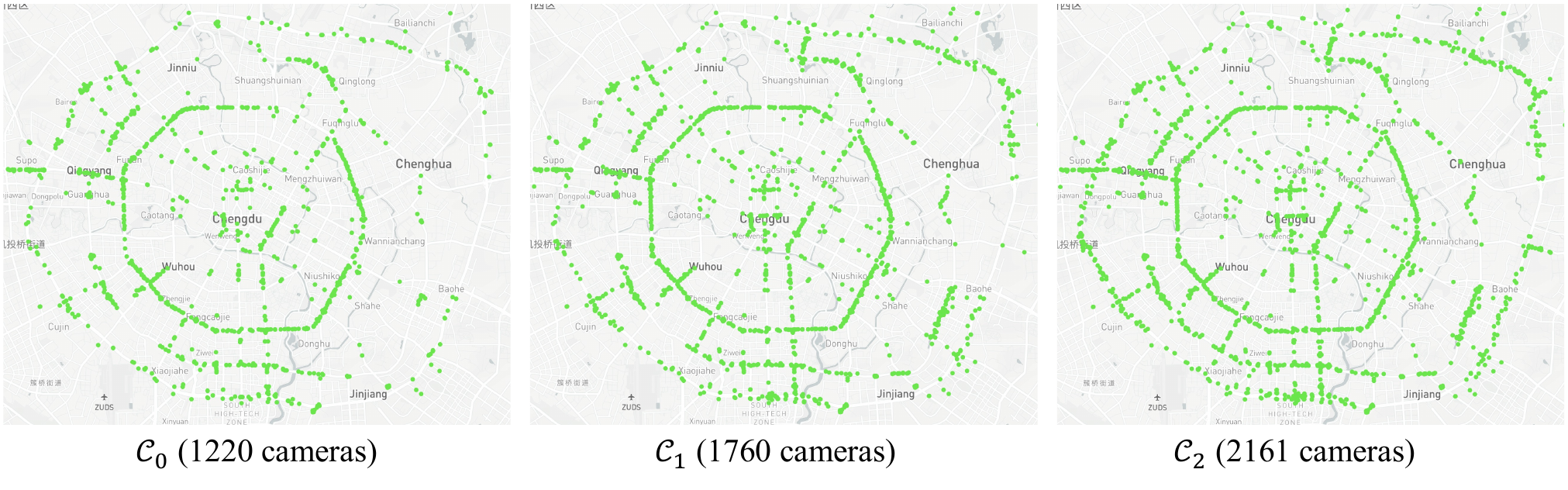}
\caption{Spatial distribution of fully functional cameras in $\mathcal{C}_i$, $i=0,1,2$.}
\label{figcamera3p}
\end{figure}

 The observed hourly traffic volumes $V_{\mathcal{C}_0}(\cdot)$, $V_{\mathcal{C}_1}(\cdot)$ and $V_{\mathcal{C}_2}(\cdot)$, corresponding to the camera sets $\mathcal{C}_0$, $\mathcal{C}_1$ and $\mathcal{C}_2$, respectively, are shown in the first row of Figure \ref{figalljan}. It can be seen that, other than the different magnitudes, the three curves are very similar in shape. To further verify this, subfigure (d) and (e) respectively show the boxplot of the point-wise ratio $V_{\mathcal{C}_1}(\cdot)/V_{\mathcal{C}_0}(\cdot)$, with a mean of $k_{0,1}=1.150$, and  $V_{\mathcal{C}_2}(\cdot)/V_{\mathcal{C}_0}(\cdot)$, with a mean of $k_{0,2}=1.266$. From these boxplots, we see that the hourly volume ratios center around their means with relatively small variations, suggesting the numerical validity of assumption {\bf (A1')}. Finally, subfigure (f) plots $V_{\mathcal{C}_0}(\cdot)$ together with rescaled $V_{\mathcal{C}_1}(\cdot)/k_{0,1}$ and $V_{\mathcal{C}_2}(\cdot)/k_{0,2}$; the three curves overlap almost entirely, suggesting that $V_{\mathcal{C}_1}(\cdot)$ and $V_{\mathcal{C}_2}(\cdot)$ can indeed be approximated as a point-wise multiple of $V_{\mathcal{C}_0}(\cdot)$.

\begin{figure}[h!]
\centering
\includegraphics[width=\textwidth]{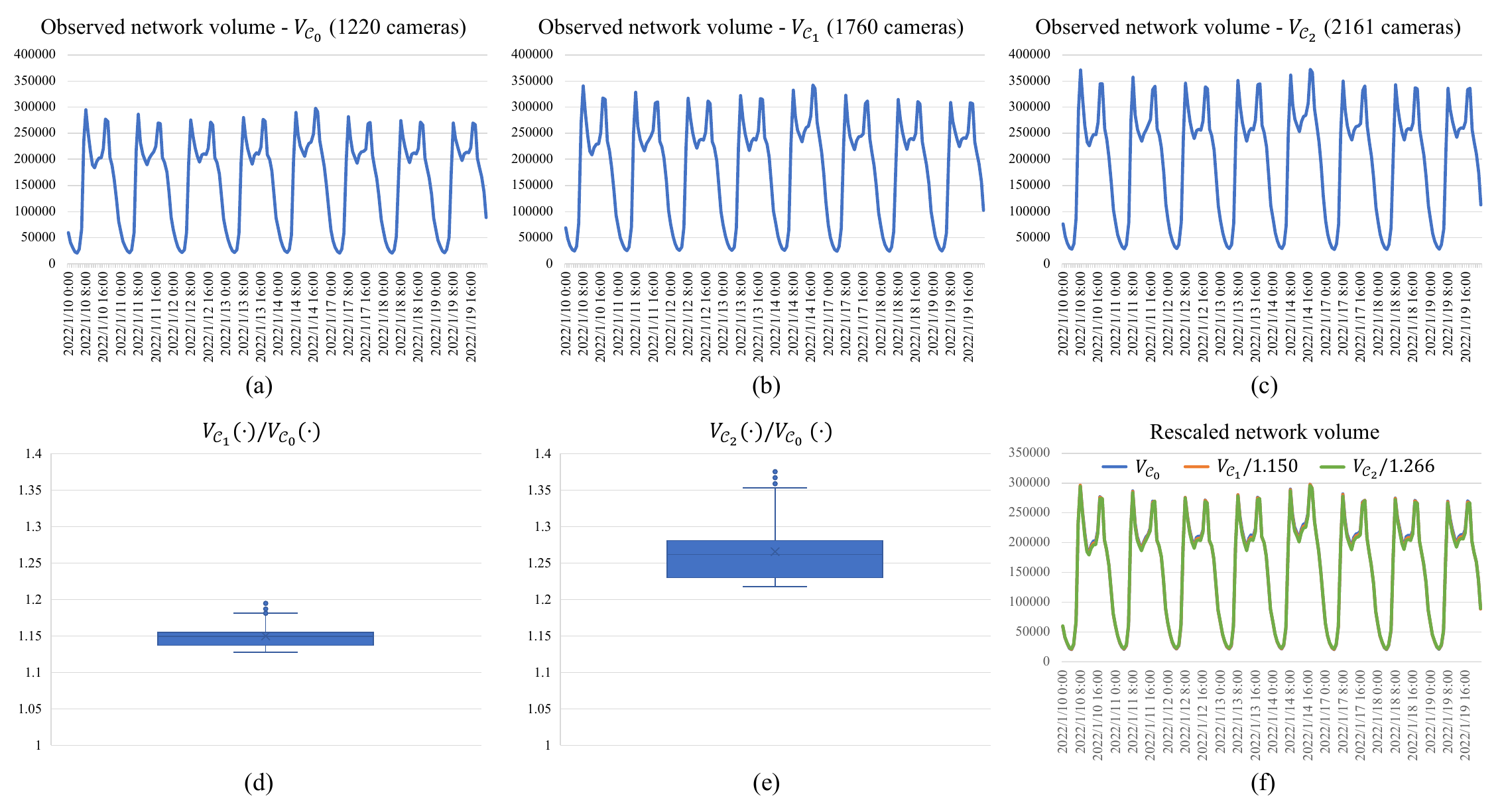}
\caption{First row: hourly network volumes during period $\mathcal{D}_2$, observed based on $\mathcal{C}_0$, $\mathcal{C}_1$ and $\mathcal{C}_2$. (d) \& (e): Box plots of the point-wise ratios $V_{\mathcal{C}_1}(\cdot)/V_{\mathcal{C}_0}(\cdot)$ \& $V_{\mathcal{C}_2}(\cdot)/V_{\mathcal{C}_0}(\cdot)$, respectively. (f): Rescaled volumes $V_{\mathcal{C}_0}(\cdot)$, $V_{\mathcal{C}_1}(\cdot)/k_{0,1}$, $V_{\mathcal{C}_2}(\cdot)/k_{0,2}$.}
\label{figalljan}
\end{figure}

\begin{remark}
The rational behind the validity of {\bf (A1')} is that the increment of observability (e.g. from extra cameras) captures extra traffic, quantified by the multiplicative factor $k_{0,1}$ or $k_{0,2}$, which is also consistent across all time intervals and multiple days, as seen in the second row of Figure \ref{figalljan}. As we derive at the end of Section \ref{secA1dis}, a condition for this to hold is that the three incremental camera sets correspond to adequate samples of the network links, and contain sufficient number of high-volume links. As shown in Figure \ref{figuniform}, links that are covered by the cameras,  or adjacent to those covered, constitute sufficient samples of all the network links, and contain the majority of 2nd and 3rd Ring Roads as well as main arterials, which have high traffic loads.
\end{remark}
	
\begin{figure}[H]
\centering
\includegraphics[width=\textwidth]{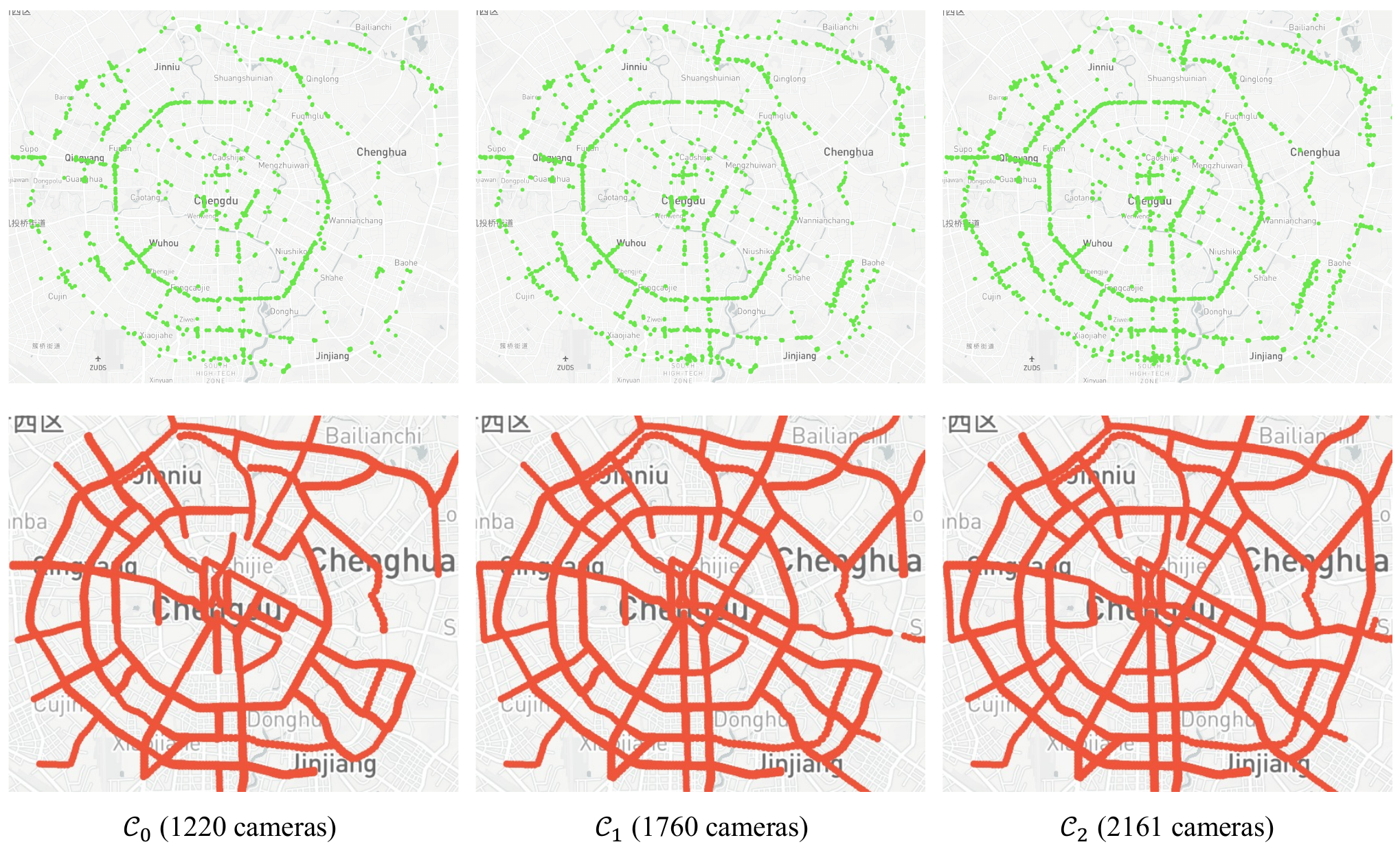}
\caption{Illustration of conditions for {\bf (A1)} and {\bf (A1')} to hold. First row: The three camera sets. Second row: Corresponding links that are either covered by the cameras or adjacent to those covered.}
\label{figuniform}
\end{figure}

\subsubsection{Observability-invariant ratio-to-critical-value}\label{subsecvalcor}
	In this section we validate the observability-invariant quantity, which is the ratio-to-critical-value (Corollary \ref{corrip}). We begin by fitting the MFD function based on three datasets from Section \ref{subsecA1}, namely the date set $\mathcal{D}_2$ with three different camera sets $\mathcal{C}_0$, $\mathcal{C}_1$ and $\mathcal{C}_2$. Similar to our findings in Section \ref{subsecnumfitting}, the vector $\tau^0=(0.4, 0.4, 0.3, 0.3)$ yields the best overall performance, and Table \ref{table2} and Figure \ref{figfitsjan} summarize the corresponding fitting performances. 
	
\begin{table}[H]
		\centering
			\caption{Goodness-of-fit performance measures for $\tau^0=(0.4, 0.4, 0.3, 0.3)$.}
			\begin{tabular}{|c|c|c|c|c|c|c|c|}
				\hline
              \multirow{2}*{Date set}  & \multirow{2}*{Observability}  & \multicolumn{3}{c|}{Original data ${\Omega _1}$} & \multicolumn{3}{c|}{Daily averaged data $\Omega_2 $}  \\
				\cline{3-8}
                                   & & StdErr & $R^2$ & SMAPE & StdErr & $R^2$ & SMAPE \\
				\hline
			 \multirow{3}{*}{$|\mathcal{D}_2|=8$} &	$|\mathcal{C}_0|=1220$ & $2476.4$ & $0.932$ & $16.33\%$& $1434.6$&$0.960$&$12.31\%$ \\
				\cline{2-8}
				& $|\mathcal{C}_1|=1760$ & 2873.6 & 0.931 & 16.54\% & 1710.5 & 0.959 & 12.27\% \\
				\cline{2-8}
				& $|\mathcal{C}_2|=2161$ & 3156.3 & 0.931 & 16.75\% & 1898.1 & 0.958 & 12.24\% \\
				\hline
			\end{tabular}
			\label{table2}
	\end{table}

\begin{figure}[H]
\centering
\includegraphics[width=\textwidth]{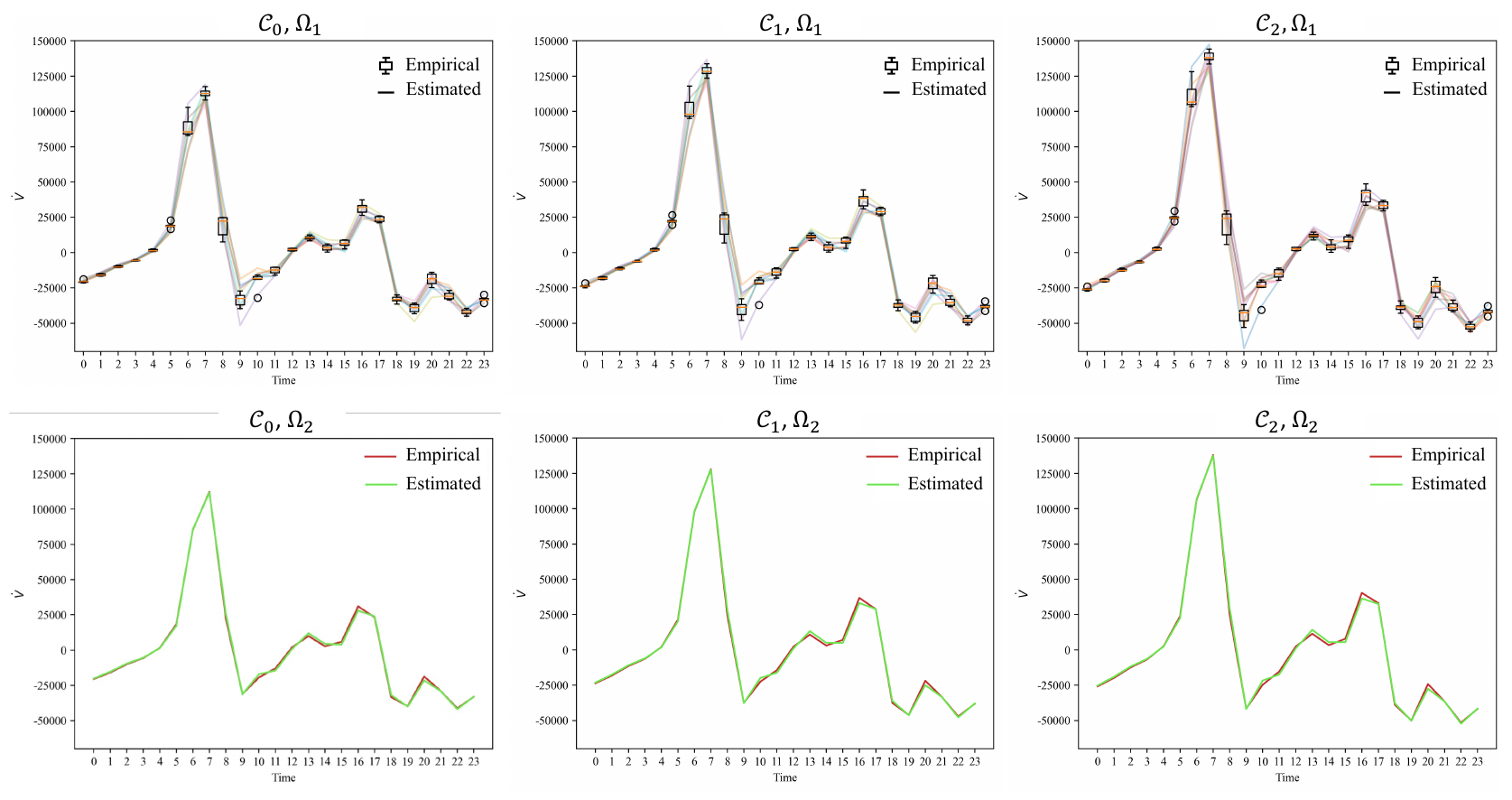}
\caption{Fitting results with three levels of observability: $\mathcal{C}_0, \mathcal{C}_1, \mathcal{C}_2$ over 8 workdays ($\mathcal{D}_2$).}
\label{figfitsjan}
\end{figure}

Next, we plot the fitted MFDs based on $\Omega_1$ and $\Omega_2$ for $|\mathcal{C}_0|=1220$, $|\mathcal{C}_1|=1760$ and $|\mathcal{C}_2|=2161$ in Figure \ref{fig6MFDs}. As the camera set is augmented, the critical values are shifted to the right, and the peaks become higher. In particular, for $\Omega_1$:
$$
V_{\mathcal{C}_0}^*=413256,\quad V_{\mathcal{C}_1}^*=475613= 1.151V_{\mathcal{C}_0}^*,\quad V_{\mathcal{C}_2}^*=516672= 1.250 V_{\mathcal{C}_0}^*.
$$
\noindent For $\Omega_2$:
$$
V_{\mathcal{C}_0}^*=433935,\quad V_{\mathcal{C}_1}^*=499697= 1.152V_{\mathcal{C}_0}^*,\quad V_{\mathcal{C}_2}^*=542926=1.251 V_{\mathcal{C}_0}^*.
$$
\noindent It can be seen that the peaks of these different MFD curves agrees nicely with the scaling factors $k_{0,1}=1.150$ (with $0.17\%$ error) and $k_{0,2}=1.266$ (with $1.26\%$ error), which are obtained from Figure \ref{figalljan}. This verifies item (2) of Corollary \ref{corrip}.

\begin{figure}[H]
\centering
\includegraphics[width=\textwidth]{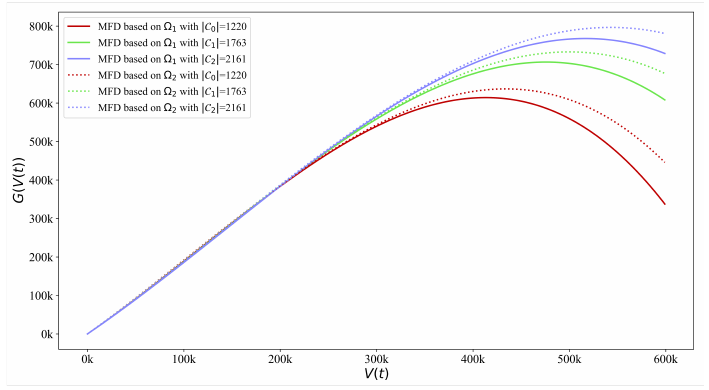}
\caption{Fitted MFD-VD based on two data-fitting sets $\Omega_1$ and $\Omega_2$, with three sets of cameras $\mathcal{C}_0$, $\mathcal{C}_1$ and $\mathcal{C}_2$.}
\label{fig6MFDs}
\end{figure}

The fact that different camera sets lead to different shapes of MFDs, as shown in Figure \ref{fig6MFDs}, is expected. This is because more cameras may capture more vehicles within an observational interval, as is the case in Figure \ref{figalljan}. In other words, the shape of the MFD is not invariant w.r.t. traffic observability, but the ratio-to-critical-value is. This is the main discovery of this work.

Notably, there is a way to remove the influence of traffic observability from the MFD functions, by performing transformation according to \eqref{Gc}. Indeed, let $G_{\mathcal{C}_i}(\cdot)$ be the unique MFD fitted based on $\mathcal{C}_i$, $i=0, 1, 2$. According to item (2) of Corollary \ref{corrip}, $G_{\mathcal{C}_0}(\cdot)$ should coincide with the transformed MFDs: 
$$
{1\over k_{0,1}}G_{\mathcal{C}_1}(k_{0,1} x)\quad  \text{and}\quad {1\over k_{0,2}}G_{\mathcal{C}_2}(k_{0,2} x)
$$
\noindent where $k_{0,1}$ and $k_{0,2}$ are such that $V_{\mathcal{C}_1}(\cdot)=k_{0,1} V_{\mathcal{C}_0}(\cdot)$, $V_{\mathcal{C}_2}(\cdot)=k_{0,2} V_{\mathcal{C}_0}(\cdot)$, per assumption {\bf (A1')}. In Figure \ref{figtransform} (top), $G_{\mathcal{C}_0}(x)$ and ${1\over k_{0,1}}G_{\mathcal{C}_1}(k_{0,1} x)$ agree completely for either fitting methods $\Omega_1$ or $\Omega_2$, as apparent from the overlapping curves. Regarding $\mathcal{C}_0$ and $\mathcal{C}_2$ in the lower subfigure, there are discernible yet very small discrepancies. This is expected for there is higher variation in $V_{\mathcal{C}_2}(t)/V_{\mathcal{C}_0}(t)$, as conveyed by the boxplots in Figure \ref{figuniform}, bringing extra noises to the fitting of the MFDs. 
	
\begin{figure}[H]
\centering
\includegraphics[width=\textwidth]{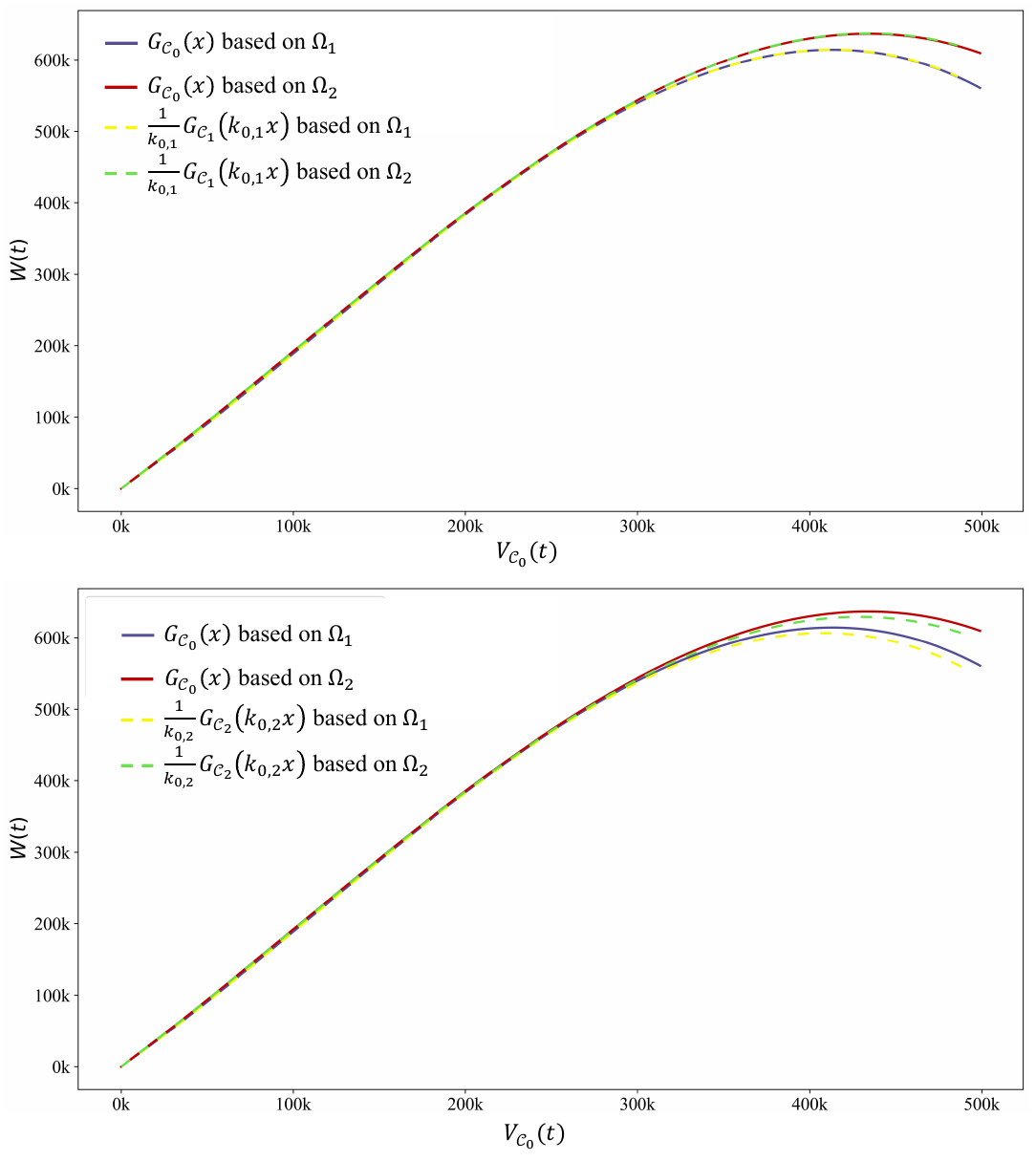}
\caption{Comparison between baseline MFD $G_{\mathcal{C}_0}(x)$ and transformed MFD ${1\over k_{0,1}}G_{\mathcal{C}_1}(k_{0,1}x)$ (top), and ${1\over k_{0,2}}G_{\mathcal{C}_2}(k_{0,2}x)$ (bottom).}
\label{figtransform}
\end{figure}

Next, Figure \ref{figkval} shows the observability-invariant ratio-to-critical-value $V_{\mathcal{C}_i}(\cdot)/ V_{\mathcal{C}_i}^*$. The left subfigure is based on $\Omega_1$, and displays the ratio-to-critical-values with $\mathcal{C}_0$, $\mathcal{C}_1$ and $\mathcal{C}_2$ over the course of 8 workdays. It can be seen that the three curves almost completely overlap. In fact, their mean absolute percentage errors (MAPEs) shown in Table \ref{fig14leftmape} are below 3\% for 24-period and 1\% for peak period (7:00-9:00, 17:00-18:00). The right subfigure focuses on $\Omega_2$, and displays the ratio-to-critical-values over 24 hours. To compare results from $\Omega_1$ and $\Omega_2$, we average $V_{\mathcal{C}_i}(\cdot)/V_{\mathcal{C}}^*(\cdot)$ from $\Omega_1$ by day, and plot them in the right subfigure. It can be seen that although the two datasets yield different MFDs (Figure \ref{fig6MFDs}), the ratio-to-critical-values are almost the same for the three observability levels (see MAPE in Table \ref{fig14rightmape}).

\begin{figure}[H]
\centering
\includegraphics[width=\textwidth]{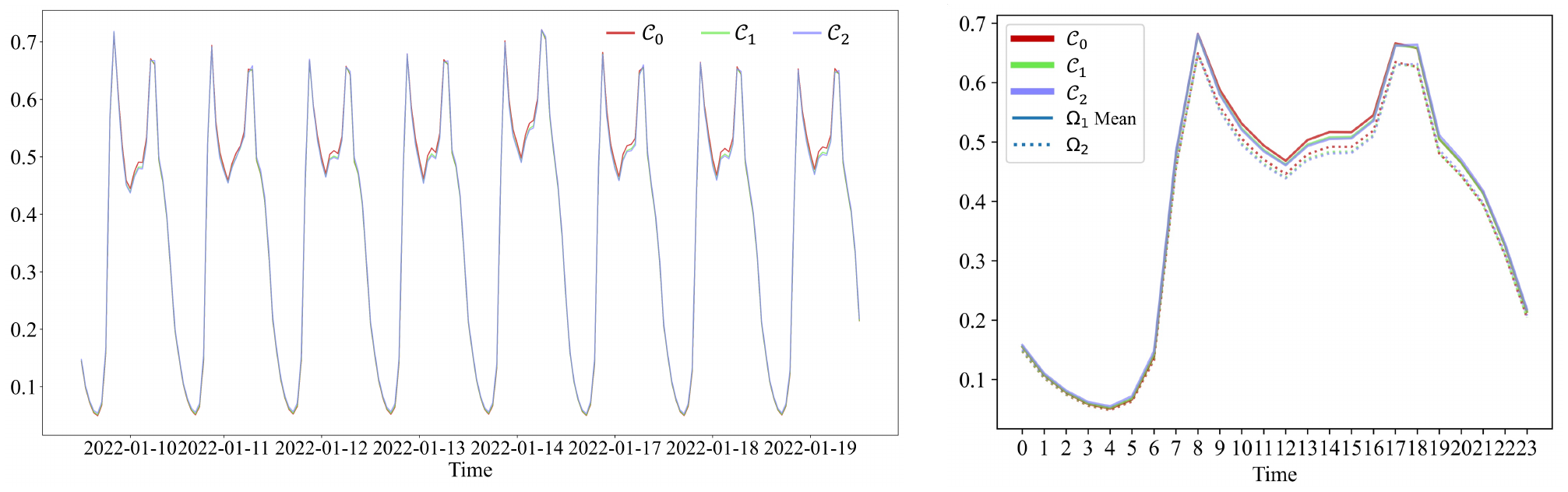}
\caption{Invariant ratio-to-critical-value. Left: $V_{\mathcal{C}_i}(\cdot) / V_{\mathcal{C}_i}^*$ based on $\Omega_1$, $i=0, 1, 2$. Right: $V_{\mathcal{C}_i}(\cdot) / V_{\mathcal{C}_i}^*$ based on $\Omega_2$ (dashed lines), compared to daily averaged $V_{\mathcal{C}_i}(\cdot) / V_{\mathcal{C}_i}^*$ from $\Omega_1$ (solid lines), $i=0,1,2$.}
\label{figkval}
\end{figure}

\begin{table}[H]
\centering
\caption{Difference of the three curves in Figure \ref{figkval} (left) in terms of MAPE. Curve $\mathcal{C}_2$ is used as baseline.}
\label{fig14leftmape}
\begin{tabular}{c|c|c}
		\hline
		&Curve & MAPE
		\\
		\hline
		\multirow{3}{*}{Whole day}& $\mathcal{C}_0$  & 2.46\%
		\\
		\cline{2-3}
		& $\mathcal{C}_1$ & 1.48\%
		\\
		\cline{2-3}
		& $\mathcal{C}_2$ & -
		\\
		\hline
		\multirow{2}{*}{Peak period}& $\mathcal{C}_0$  & 0.64\%
		\\
		\cline{2-3}
		\multirow{2}{*}{(7-9:00, 17-18:00)}& $\mathcal{C}_1$ & 0.36\%
		\\
		\cline{2-3}
		& $\mathcal{C}_2$ & -
		\\
		\hline
	\end{tabular}
\end{table}

\begin{table}[H]
\centering
\caption{Difference of the six curves in Figure \ref{figkval}  (right) in terms of MAPE. Curve $\mathcal{C}_2$ ($\Omega_2$) is used as baseline.}
\label{fig14rightmape}
\begin{tabular}{c|c|c|c|c}
		\hline
		&Curve& MAPE&Curve& MAPE
		\\
		\hline
		\multirow{3}{*}{Whole day}& $\mathcal{C}_0$($\Omega_2$)  & 2.43\% & $\mathcal{C}_0$($\Omega_1$ Mean)  & 4.51\%
		\\
		\cline{2-5}
		& $\mathcal{C}_1$($\Omega_2$) & 1.47\%& $\mathcal{C}_1$($\Omega_1$ Mean) & 3.87\%
		\\
		\cline{2-5}
		& $\mathcal{C}_2$($\Omega_2$) & - & $\mathcal{C}_2$($\Omega_1$ Mean) & 5.08\%
		\\
		\hline
		\multirow{2}{*}{Peak period}& $\mathcal{C}_0$($\Omega_2$)  & 0.56\%& $\mathcal{C}_0$($\Omega_1$ Mean)  & 5.01\%
		\\
		\cline{2-5}
		\multirow{2}{*}{(7-9:00, 17-18:00)}& $\mathcal{C}_1$($\Omega_2$) & 0.34\% & $\mathcal{C}_1$($\Omega_1$ Mean) & 4.81\%
		\\
		\cline{2-5}
		& $\mathcal{C}_2$($\Omega_2$) & -& $\mathcal{C}_2$($\Omega_1$ Mean) & 5.08\%
		\\
		\hline
	\end{tabular}
\end{table}

Finally, we acknowledge the fact that all the results presented so far are based on the chosen Road Congestion Indices from Table \ref{tabRCI}. To understand the impact of such choices on our final results, we consider two more sets of RCI values, respectively shown as `Low' and `High' in Table \ref{tabRCIs}.  

\begin{table}[h]
\centering
\caption{Different numerical settings of the RCIs}
\label{tabRCIs}
\begin{tabular}{c|c|c|c|c}
\hline
      Traffic  status & Free-flow & Slow  & Congested & Heavily congested
    	\\
    	\hline
	RCI interval & $[1.0,\, 1.5)$  & $[1.5,\, 2.0)$ & $[2.0,\, 4.0)$ & $[4.0,\, \infty)$  
        \\
        \hline
        Low (lower 25\%) & 1.125 & 1.625 & 2.5 & 4.25 
        \\
        \hline
	Med (50\%)   & 1.250 & 1.750  & 3.0 & 5.00
    	\\
    	\hline
	High (upper 25\%) & 1.375 & 1.875 & 3.5 & 6.00
	\\
    	\hline
\end{tabular}
\end{table}

We note that the network congestion index $D(t)$ and free-flow time $\tau^0$ appear in the form of a product in the ODE \eqref{VCODE}, which means the fitted MFD, as well as ratio-to-critical-values, depend only on their product, not individual values. For different choices in Table \ref{tabRCIs}, the resulting $D(\cdot)$'s, which are averaged over the 8 workdays, are shown in Figure \ref{figLMH}(a). Note that they may correspond to different $\tau^0$ in the fitting process, e.g. higher $D(t)$ corresponds to lower $\tau^0$, but their product $\tau^0 D(t)$ is supposed to be the same. Table \ref{tabR2sforLMH} shows the fitting performance (in terms of $R^2$) with the low, median and high $D(t)$, under the same $\tau^0=(0.4, 0.4, 0.3, 0.3)$\footnote{The reason for the same $\tau^0$ is that we set the increment of $\tau^0_i$ to be 0.1 hr, and $(0.4, 0.4, 0.3, 0.3)$ remain the optimal fitting value for all three $D(t)$'s.}.

\begin{table}[h]
\centering
\caption{$R^2$ of the MFD function fitting with low, median and high $D(\cdot)$.}
\label{tabR2sforLMH}
\begin{tabular}{c|c|c|c}
\hline
      & $\mathcal{C}_0$ & $\mathcal{C}_1$  & $\mathcal{C}_2$ 
        \\
        \hline
        Low  & 0.936 & 0.935 & 0.934
        \\
        \hline
	Med   & 0.932 & 0.931  & 0.931 
    	\\
    	\hline
	High & 0.927 & 0.927 & 0.926 
	\\
    	\hline
\end{tabular}
\end{table}

Figures \ref{figLMH}(b)-(d) show the ratio-to-critical-values corresponding to the low, median and high $D(t)$. It can be seen that the ratio-to-critical-values are very similar for the three cases, and are slightly higher for larger $D(t)$. Table \ref{fig15mape} further shows the difference (MAPE) between the individual curves of Figures \ref{figLMH}(b)-(d). Their differences are negligible, which means the results are not sensitive to the different choices of RCIs.

\begin{figure}[H]
\centering
\includegraphics[width=\textwidth]{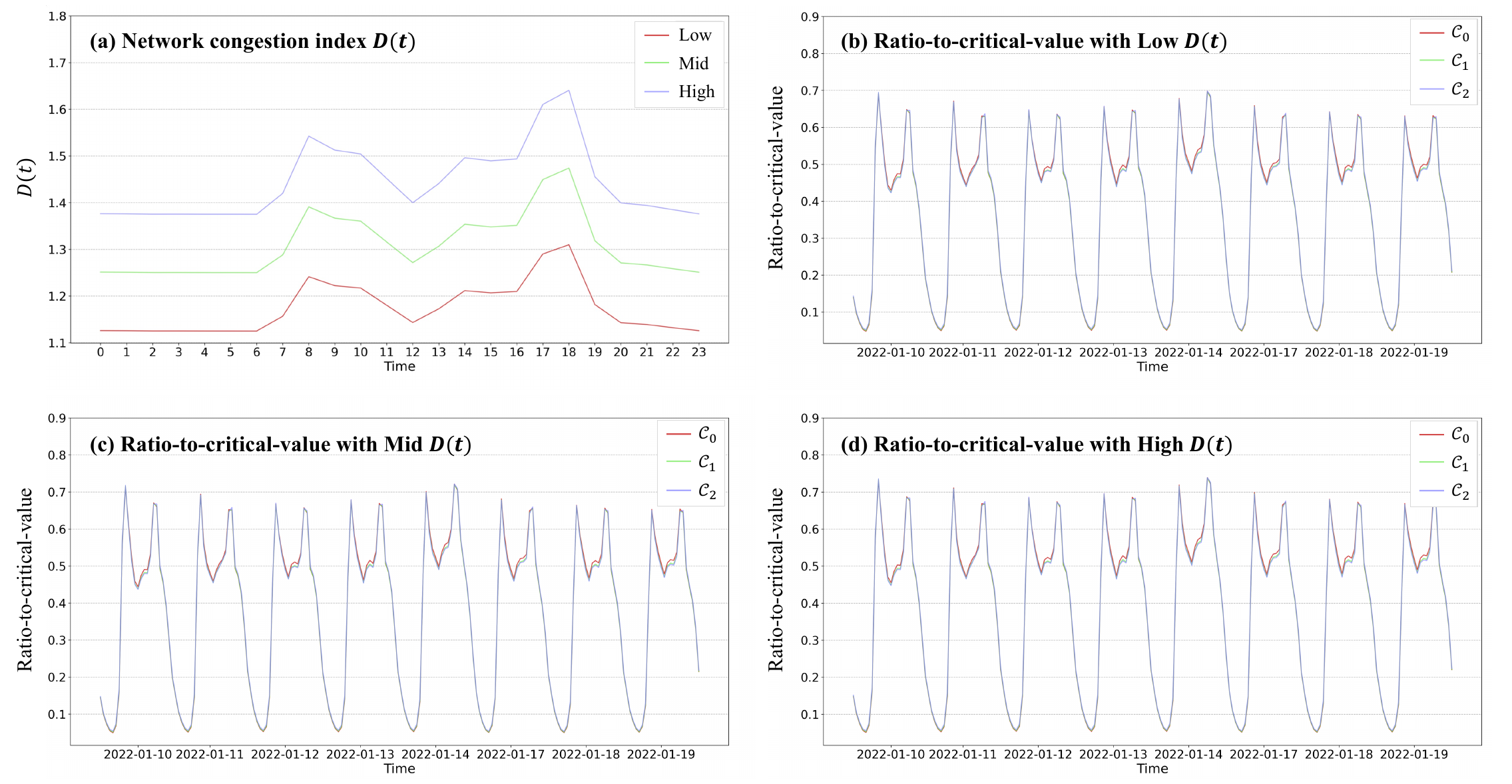}
\caption{Sensitivity of ratio-to-critical-value to RCI values.}
\label{figLMH}
\end{figure}

\begin{table}[H]
	\centering
	\caption{Difference of the three curves in \ref{figLMH} (b)-(d) in terms of MAPE. Curve $\mathcal{C}_2$ is used as baseline.}
	\label{fig15mape}
	\begin{tabular}{c|c|c|c|c}
		\hline
		& & \multicolumn{3}{c}{MAPE}
		\\
		\cline{1-5}
		&Curve & (b)  &   (c) &   (d)
		\\
		\hline
		\multirow{3}{*}{Whole day}&$\mathcal{C}_0$ &  $2.47\%$& $ 2.46\%$& $2.44\%$
		\\
		\cline{2-5}
		& $\mathcal{C}_1$ & $1.50\%$&  $ 1.48\% $ & $ 1.47\% $
		\\
		\cline{2-5}
		& $\mathcal{C}_2$ & - &  - & -
		\\
		\hline
		\multirow{2}{*}{Peak period}&$\mathcal{C}_0$  & 0.63\%& 0.64\%& 0.72\%
		\\
		\cline{2-5}
		\multirow{2}{*}{(7-9:00, 17-18:00)} & $\mathcal{C}_1$ & 0.40\%&  0.36\%& 0.35\%
		\\
		\cline{2-5}
		& $\mathcal{C}_2$ & - & - &  -
		\\
		\hline
	\end{tabular}
\end{table}

\subsection{The invariance property with different time granularity}\label{subsectimegra}

In this section, we assess the impact of smaller time step size $\Delta s$ (i.e. 30min, 20min, 10min) on the MFD-VD model and the invariance property. The key expectation is that the MFD and ratio-to-critical-value should be robust against smaller time steps employed. Another concern is that the previous $\Delta s=60$min could have averaged out some transient dynamics, which will be analyzed in depth here. Throughout this section, the camera set is chosen to be $|\mathcal{C}_2|=2161$.

We begin with the network volume $V^{\Delta s}_{\mathcal{C}_2}(t_s)$ with $\Delta s=$ 60min, 30min, 20min and 10min, as shown in Figure \ref{figVDs}(a). It can be seen that the resulting time series are highly correlated. In fact, by multiplying appropriate values $k_{\text{30min}}=1.6$, $k_{\text{20min}}=2.1$, $k_{\text{10min}}=3.2$ to the corresponding time series $V^{\Delta s}_{\mathcal{C}_2}$, we recover almost identical curves, as visualized in Figure \ref{figVDs}(b) and analyzed in Table \ref{fig16(b)mape}. Several observations are made:
\begin{itemize}
\item[(1)] The traffic dynamics are relatively stable within the 1-hr period, because the volume curves with various $\Delta s$ are linearly correlated; 
\item[(2)] The differences (MAPEs) increase with higher temporal resolution (smaller $\Delta s$), indicating the presence of complex dynamics at smaller time scales in the network;
\item[(3)] $k_{\Delta s}$ is considerably smaller than 60min/$\Delta s$ (e.g. $k_{30\text{min}}=1.6< 60/30=2$, $k_{10\text{min}}=3.2<60/10=6$). This is a manifestation of traffic occupying multiple periods, or cars not captured by cameras in certain periods. More on such phenomenon will be pursued in future work. 
\end{itemize}

\begin{figure}[H]
\centering
\includegraphics[width=\textwidth]{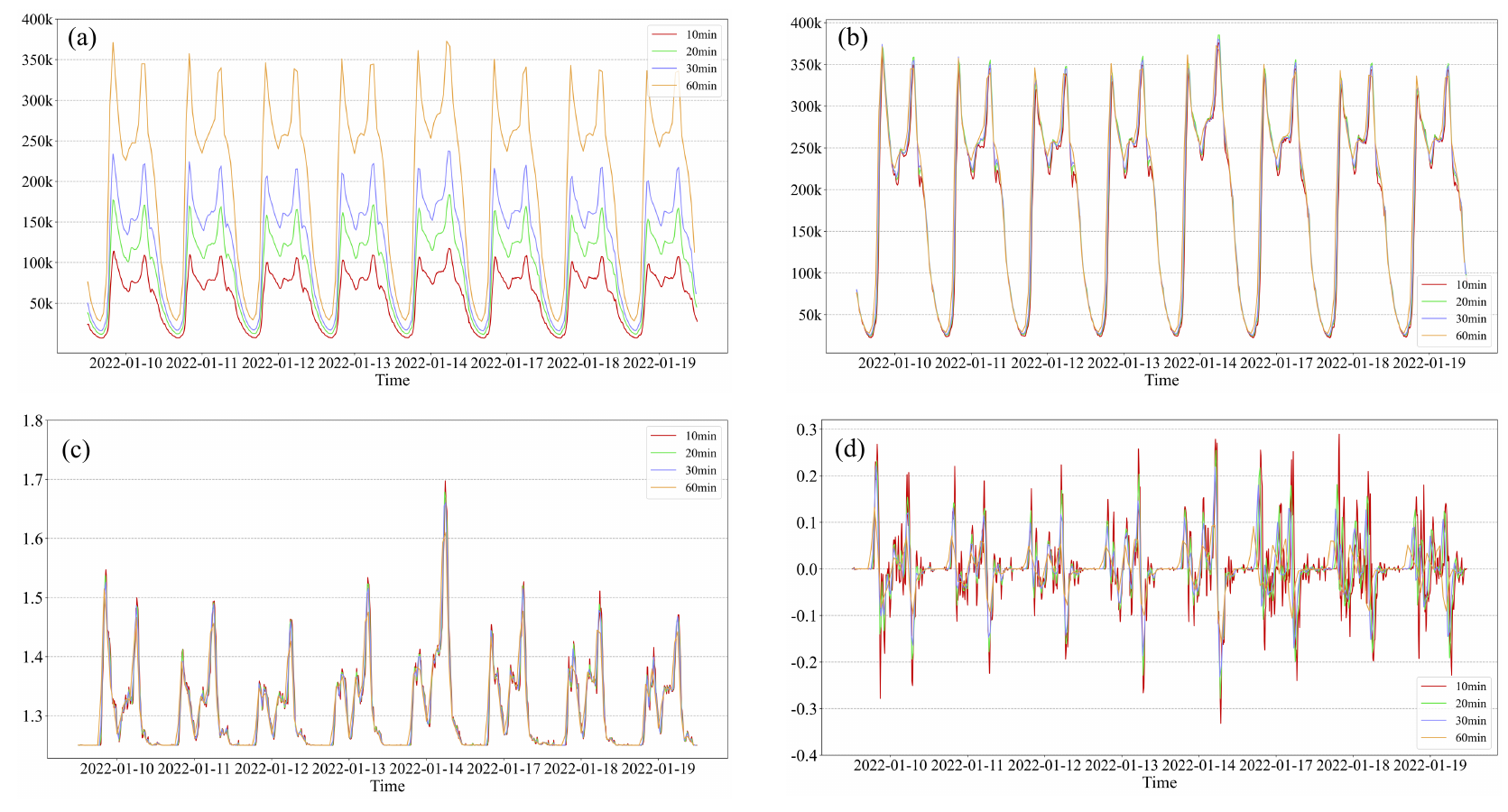}
\caption{Volume and delay with different time steps. (a) Observed network volume; (b) Observed network volume rescaled with a multiplicative constant; (c) Network congestion index $D(\cdot)$; (d) Time derivative of congestion index $\dot D(\cdot)$.}
\label{figVDs}
\end{figure}

\begin{table}[h]
	\centering
	\caption{Difference of the four curves in Figure \ref{figVDs} (b)-(c) in terms of MAPE. Curve `60 min' is used as baseline.}
	\label{fig16(b)mape}
	\begin{tabular}{c|c|c|c}
		\hline
		& & \multicolumn{2}{c}{MAPE}
		\\
		\cline{1-4}
		&Curve& (b) & (c)
		\\
		\hline
		\multirow{4}{*}{Whole day}& $10$ min  & 9.77\% &  1.10\%
		\\
		\cline{2-4}
		& $20$ min & 7.44\%&  0.90\%
		\\
		\cline{2-4}
		& $30$ min & 6.01\%&  0.67\%
		\\
		\cline{2-4}
		& $60$ min & - & -
		\\
		\hline
		\multirow{3}{*}{Peak period}& $10$ min  & 7.86\%&  2.03\%
		\\
		\cline{2-4}
		\multirow{3}{*}{(7-9:00, 17-18:00)}& $20$ min & 5.40\%&  1.65\%
		\\
		\cline{2-4}
		& $30$ min & 4.05\%&  1.25\%
		\\
		\cline{2-4}
		& $60$ min & - &  -
		\\
		\hline
	\end{tabular}
\end{table}

Next, the network delay $D(s)$ and $\dot D(s)$ for $\Delta s=60$min, 30min, 20min and 10min are shown in Figure \ref{figVDs}(c) \& (d), respectively. Compared to network volume, the delay $D(\cdot)$ undergoes significant variations at smaller temporal scales, which renders even higher fluctuations in the derivative $\dot D(\cdot)$. This indicates the existence of transient yet considerable congestion at the network level. Such variations of congestion at fine temporal scales could bring extra noises to the multilinear regression of the MFD function; indeed, we have
$$
R^2 = 0.916~ \text{(30 min)}, \quad R^2 = 0.901~ \text{(20 min)},\quad R^2 =  0.861 ~\text{(10 min)},
$$
decreasing with the time step size. 

\begin{figure}[H]
\centering
\includegraphics[width=\textwidth]{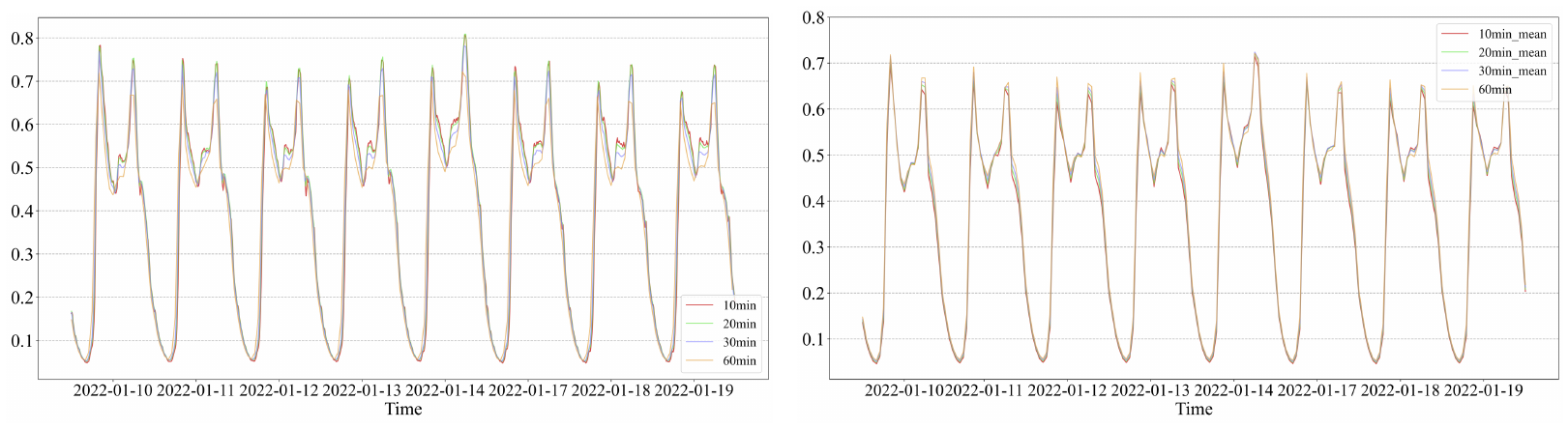}
\caption{Left: Ratios-to-critical-values with different time step sizes; Right: Hourly averages of the ratios-to-critical-values with different time step sizes.}
\label{figRTCVdelta}
\end{figure}

\begin{table}[h]
	\centering
	\caption{Difference of the four curves in Figure \ref{figRTCVdelta} in terms of MAPE. Curve `60 min' is used as baseline.}
	\label{fig17mape}
	\begin{tabular}{c|c|c|c|c}
		\hline
		&\multicolumn{2}{c|}{Left: Different step sizes}&\multicolumn{2}{c}{Right: Same step size (60 min)} 
		\\
		\hline
		&Curve & MAPE &Curve& MAPE
		\\
		\cline{1-5}
		\multirow{4}{*}{Whole day}& 10min & 12.19\%& 10min\_mean& 6.57\%
		\\
		\cline{2-5}
		& 20min& 10.86\%& 20min\_mean&  4.31\%
		\\
		\cline{2-5}
		& 30min& 8.00\%& 30min\_mean&  2.56\%
		\\
		\cline{2-5}
		& 60min & -& 60min& -
		\\
		\hline
		\multirow{3}{*}{Peak period}& 10min & 7.72\%& 10min\_mean&  3.66\%
		\\
		\cline{2-5}
		\multirow{3}{*}{(7-9:00, 17-18:00)}& 20min & 6.39\%& 20min\_mean &  2.25\%
		\\
		\cline{2-5}
		& 30min & 4.84\%& 30min\_mean&  1.27\%
		\\
		\cline{2-5}
		& 60min & - & 60min&  -
		\\
		\hline
	\end{tabular}
\end{table}

Finally, Figure \ref{figRTCVdelta} (left) shows the ratios-to-critical-values corresponding to various time step sizes. The higher delays with finer temporal resolution, as depicted in Figure \ref{figVDs}(c), render an overestimation of the ratios-to-critical-values (especially $20$ min and $10$ min). As shown in Table \ref{fig17mape}, the MAPE relative to the `60 min' curve grows with smaller time intervals. Nevertheless, when we average these quantities to the same temporal scale (60 min), the resulting shapes are almost identical, as shown in Figure \ref{figRTCVdelta} (right) and Table \ref{fig17mape} (with MAPE below 6.57\%). Such similarities can be explained via the concept of {\it temporally partial} observation and the invariance property. Indeed, if we consider 1-hr vehicle count as full observation, then a 20-min count within that hour is only partial observation. Figure \ref{figVDs}(b) suggests that a condition analogous to assumption {\bf (A1')} holds in this case:

\vspace{0.05 in}
\noindent {\bf Assumption (A2)} Let time step $\Delta s$  be a divider of 60min, then $V^{\text{60min}}_{\mathcal{C}_2}(t_{s'})=k_{\Delta s}\cdot V^{\Delta s}_{\mathcal{C}_2}(t_s)$, $\forall t_s$ such that $[t_s,\,t_{s+1})\subset [t_{s'},\,t_{s'+1})$, for some $k_{\Delta s}\geq 1$ that depends only on $\Delta s$.
\vspace{0.05 in}

\noindent Using {\bf (A2)}, we can easily show the invariance property, provided that the delay $D$ is unchanged with smaller time steps. Nevertheless, the actual delays corresponding to smaller step sizes are higher, especially at some time instances during traffic peaks, indicating higher congestion than what the 1-hr version has reckoned. Therefore,  the ratios-to-critical-values with smaller step sizes are higher.

\section{Conclusion}\label{secConclude}

This paper presents a macroscopic fundamental diagram (MFD) model with a volume-delay relationship (MFD-VD), using license plate cameras (LPCs) and road congestion indices (RCIs). The LPCs are used to calculate (partial) network accumulation (volume) while the RCIs are used to estimate average network delays. The key findings of this paper are as follows.

\begin{itemize}
\item[(1)] Based on volume LPC and RCI data, an ordinary differential equation (ODE), involving volume $V(\cdot)$ and delay $D(\cdot)$, is derived based on flow conservation, flow propagation and an accumulation-based MFD on a network level. 

\item[(2)] Building on such an ODE, the empirical data fitting of the MFD function amounts to a multilinear regression parameterized by $\tau^0$, which can be readily performed and evaluated. The goodness-of-fit performances for multi-day ($\Omega_1$) and single-day ($\Omega_2$) fitting are $(R^2,\,\text{SMAPE})=(0.93,\,17\%)$ and $(R^2,\,\text{SMAPE})=(0.96,\,12\%)$ (Table \ref{table2}). Furthermore, the impact of $\tau^0$ on the goodness-of-fit appears nonlinear, which should be further investigated numerically. 

\item[(3)] Assuming {\bf (A1)} or {\bf (A1')} holds, the quantity of ratio-to-critical-value, which is an indicator of network saturation and efficiency, is invariant w.r.t. the level of observability afforded by the LPCs. This means that, even though the true network volume cannot be fully observed, working with a proper set of cameras could still estimate such an important quantity. 

\item[(4)] For {\bf (A1)} or {\bf (A1')} to hold, the road links either covered by the LPCs, or adjacent to those covered, should be evenly distributed in the network and contain as many high-volume links as possible. Such conditions have been verified using real-world data, which also prove the validity of {\bf (A1)} or {\bf (A1')}. 

\item[(5)] Compared to conventional data sources such as loop detector data and floating car data, the proposed framework does not require full observability of the network accumulation nor {\it a priori} known proportion of detected vehicles. In addition, uniform distribution of the cameras is also unnecessary; see (4).

\item[(6)] The ratio-to-critical-value is also invariant w.r.t. the time step size. This is attributed to the a similar invariance principle with temporally partial observation. According to our empirical study, an assumption analogous to {\bf (A1')} holds, which concerns with network volume calculated using different time step sizes. 
\end{itemize}

As future search, the assumptions {\bf (A1)} or {\bf (A1')} will be further investigated regarding the possibly time-dependent scaling factor $k$. Moreover, the volume-delay relationship will be analyzed in various forms of spatial and temporal aggregation to analyze network performance in the volume-delay domain. As a final note, the network volume $V_{\mathcal{C}}(t_s)$ obtained from camera set $\mathcal{C}$ during certain time interval $[t_s,\,t_{s+1})$ displays remarkable scalability, both spatially (i.e. with different camera sets) and temporally (i.e. with different time step sizes). This is a unique characteristic of license plate camera data, and is the foundation of the invariance property derived in this paper. Future research will explore more applications of this feature.

As a data-driven study, the applicability and validity of our model  depends on the availability and quality of data. In addition to the aforementioned conditions regarding LPC and RCI data, which are essential, real-world issues such as sensor or transmission failure could compromise the performance of our model. Robust analysis against these realistic scenarios will be pursued in the future.

\section*{Acknowledgement}
This work is supported by the National Natural Science Foundation of China through grants 72071163 and 72101215, and the Natural Science Foundation of Sichuan Province through grants 2022NSFSC0474.

\section*{Appendix. Validation of Eqn \eqref{eqnproxi}}

We start by calculating the right hand side of \eqref{eqnproxi} for every interval $[t, t+\Delta t]$:
\begin{equation}\label{eqna1}
 {{1\over \sum_{r\in R}{\omega_r}}\sum_{r\in R}\omega_r\sum_{l\in r}\text{FTT}_l\cdot \Phi_l\big(t_l,\,t_l+\tau_l\big)\over {1\over\sum_{r\in R}{\omega_r}}\sum_{r\in R}\omega_r\sum_{l\in r}\text{FTT}_l}
\end{equation}
An exact method based on straightforward bookkeeping is devised for calculating route travel time for $r=(l_1,\ldots, l_n)$ with departure time $t_{l_1}$, where the link travel time is given as $\text{FTT}_l\cdot \Phi(t_l,  t_l+\tau_l)$. Then, the route choices as well as trip distributions are treated via a simulation experiment. Specifically, within the target network we generate large number of trips (over 1.7 million) with random origin-destination pairs and random departure times within $[t, t+\Delta t]$ (the shortest route under the free-flow condition is used as default). Then, both the numerator and denominator in \eqref{eqna1} are calculated as the average over these large number of random trips. 

The quantity shown in \eqref{eqna1} is compared with $\Phi(t, t+\Delta t)$ in terms of mean absolute percentage error (MAPE) to measure the goodness of approximation. Figure \ref{figApp1} shows the MAPEs for $t=0,\ldots, 23$ and $\Delta t=1$ hr, where the box plots summarize the results from six different days. It can be seen that most of the errors are below 3\%, with a few higher ones in the morning and afternoon peaks. This shows the validity of the proposed approximation, namely Eqns \eqref{eqndass0} and \eqref{eqnproxi}.

\begin{figure}[H]
\centering
\includegraphics[width=.6\textwidth]{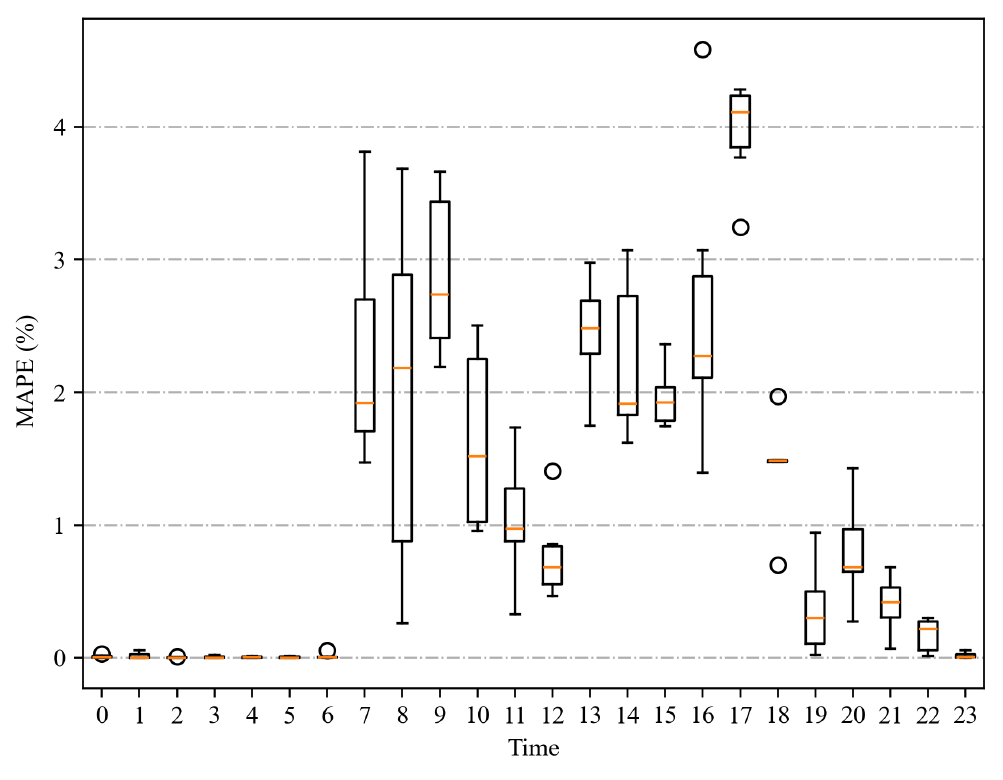}
\caption{Time-varying MAPEs between $\Phi(t, t+\Delta)$ and the quantity in \eqref{eqna1} estimated from simulations.}
\label{figApp1}
\end{figure}

\end{document}